\documentclass[12pt]{iopart}
\usepackage{latexsym}
\usepackage[T1]{fontenc}
\usepackage[utf8x]{inputenc}
\usepackage{times}
\usepackage{float}
\usepackage{amsthm}
\usepackage{amssymb}
\usepackage{amsfonts}
\usepackage{amstext}
\usepackage{graphicx}
\usepackage{array}
\usepackage{multirow}
\usepackage{subfig}
\usepackage{float}
\usepackage{indentfirst}
\usepackage{enumitem}
\usepackage{bbm}
\usepackage{multirow}
\usepackage{color}
\usepackage{mathabx}
\usepackage{cite}

\def\mc {\mathcal}
\def\mk {\mathfrak}

\def\RR {\mathbb{R}}
\def\CC {\mathbb{C}}
\def\PP {\mathbb{P}}
\def\EE {\mathbb{E}}
\def\ZZ {\mathbb{Z}}

\def \PPH {\PP(\mc{H})}
\def\mH {\mc{H}}
\def\mF {\mc{F}}
\def\mC {\mc{C}}
\def\mP {\mc{P}}

\def\mkgl {\mk{gl}}
\def\mkg {\mk{g}}
\def\mkk {\mk{k}}
\def\mku {\mk{u}}
\def\mkt {\mk{t}}
\def\kpsi {\ket{\Psi}}
\def\bpsi {\bra{\Psi}}
\def\klambda {\ket{\lambda}}
\def\blambda {\bra{\lambda}}
\def\bone {\mathbbm{1}}
\def \i {\iota}
\def \poly {\mc{P}(\mH)}
\def \vac {\ket{\Omega}}

\def\SpC#1{\bigg{\langle}#1\bigg{\rangle}_\CC}
\def\conv#1{{\rm Conv}\left(#1\right)}

\newcommand{\bra}[1]{\mbox{$\langle #1 |$}}
\newcommand{\ket}[1]{\mbox{$| #1 \rangle$}}
\newcommand{\bk}[2]{\ensuremath{\langle #1 | #2 \rangle}}
\newcommand{\kb}[2]{\ensuremath{| #1 \rangle\!\langle #2 |}}

\newcommand{\spC}[1]{\mbox{$\langle #1 \rangle_\CC$}}
\newcommand{\spR}[1]{\mbox{$\langle #1 \rangle_\RR$}}
\newcommand{\floor}[1]{\ensuremath\lfloor #1\rfloor}

\newtheorem{theorem}{Theorem}
\newtheorem{lemma}[theorem]{Lemma}

\newtheorem{definition}[theorem]{Definition}
\newtheorem{remark}{Remark}

\begin{document}

\title{Quantum marginals from pure doubly excited states}

\author{Tomasz Maci\k{a}\.{z}ek}
\address{Center for Theoretical Physics, Polish Academy of
Sciences, Al. Lotnik\'ow 32/46, 02-668 Warszawa, Poland}

\author{Valdemar Tsanov}
\address{Mathematisches Institut, Universit\"at G\"ottingen, Bunsenstra{\ss}e 3-5, 37073 G\"ottingen, Germany}
\date{\today}

\begin{abstract}
The possible spectra of one-particle reduced density matrices that are compatible with a pure multipartite quantum system of finite dimension form a convex polytope. We introduce a new construction of inner- and outer-bounding polytopes that constrain the polytope for the entire quantum system. The outer bound is sharp. The inner polytope stems only from doubly excited states. We find all quantum systems, where the bounds coincide giving the entire polytope. We show, that those systems are: i) any system of two particles ii) $L$ qubits, iii) three fermions on $N\leq 7$ levels, iv)  any number of bosons on any number of levels and v) fermionic Fock space on $N\leq 5$ levels. The methods we use come from symplectic geometry and representation theory of compact Lie groups. In particular, we study the images of proper momentum maps, where our method describes momentum images for all representations that are spherical.
\end{abstract}

\maketitle

\section{Introduction}\label{sec:introduction}
The quantum marginal problem is the problem of describing the set of all possible reduced density matrices that correspond to some pure state of a system of many particles. This is a difficult and fundamental question, which occurs in many branches of physics and quantum chemistry. In 1995 it was designated by National Research Council of USA as one of ten most prominent research challenges in quantum chemistry \cite{raport}. Let us next briefly describe the main points of this problem.  A pure quantum state can be described by a normalised vector from a Hilbert space. Equivalently, in the density matrix formalism a pure quantum state is a rank-one projector on the corresponding vector. However, one usually does not need the whole density matrix in order to describe some properties of a quantum system and the knowledge of reduced density matrices is sufficient. The reduced density matrix associated to a subsystem $A'$ of the entire quantum system $A$ is defined by averaging the whole density matrix $\rho=\kb{\Psi}{\Psi}
$ over the subsystem $A-A'$. The average is done via computing the partial trace over  the Hilbert space corresponding to $A-A'$, which we denote by $\rho_{A'}=\tr_{A-A'}\rho$. In general, one can consider reduced density matrices that correspond to subsystems that overlap, i.e. $\rho_{A'}$ and $\rho_{A''}$, where $A'\cap A''\neq\emptyset$. Then, $\rho_{A'}$ and $\rho_{A''}$ are called overlapping marginals. Such marginals occur naturally in quantum chemistry. For example, while calculating energy of a system with a finite number of particles, where the interactions occur only in pairs of particles, the sufficient information is contained in the two-particle reduced density matrices. Indeed, for 
such a system the Hamiltonian is a sum of one- and two-particle 
Hamiltonians, $H=\sum_i H_i+\sum_{i,j} H_{i,j}$. Furthermore, if the quantum system is in a state $\ket{\Psi}$, then the density matrix reads $\rho=\kb{\Psi}{\Psi}$ and the energy is given by
\[\fl E=\tr\left(H\rho\right)=\sum_{i}\tr\left(H_{i}\rho\right)+\sum_{i,j}\tr\left(H_{i,j}\rho\right)=\sum_{i}\tr\left(H_{i}\rho_{i}\right)+\sum_{i,j}\tr\left(H_{i,j}\rho_{i,j}\right),\]
where $\rho_i$ is the one-particle reduced density matrix for $i$th particle and $\rho_{i,j}$ is the two-particle reduced density matrix for the subsystem of $i$-th and $j$-th particle. While modelling molecules, the stable electronic configuration is described by a state that minimises the energy. In other words, we are interested in finding the following minimum:
\[\fl E_0=\min_{\rho=\kb{\Psi}{\Psi}}\tr\left(H\rho\right)=\min_{\{\rho_i=\rm{1PRDM},\rho_{i,j}=\rm{2PRDM}\}}\left(\sum_{i}\tr\left(H_{i}\rho_{i}\right)+\sum_{i,j}\tr\left(H_{i,j}\rho_{i,j}\right)\right).\]
Such a problem can in principle be solved numerically. However, taking the domain of minimisation to be the entire space of pure states is very inefficient, as the dimension of such a space grows exponentially with the number of particles. The space of the one- and two-particle reduced density matrices ($\rm{1PRDM},\ \rm{2PRDM}$) is much smaller, therefore an algorithm that uses the optimal domain of minimisation would be much more robust. However, the problem of describing the set of two-particle reduced density matrices seems to be intractable and there are no solutions known, even for low-dimensional systems. From the point of view of computational complexity, the problem of deciding whether a set of two-particle density matrices is compatible with some pure $N$ particle state is QMA complete, i.e. is expected to be intractable even for a quantum computer \cite{LCV,Liu}. One way to simplify this problem is to approximate the original hamiltonian by a sum of one-particle hamiltonians, as, for example, in the Hartree-
Fock method. Then,
 the minimisation is done over the set of one-particle reduced density matrices, which is easier to describe. The problem of deciding compatibility of a set of one particle reduced density matrices with some pure state is called the one-body quantum marginal problem and it is NP-hard \cite{BCMW15}. This problem is computationally simpler mainly because the one-body quantum marginals are non-overlapping. In this paper, we only consider the one-body quantum marginal problem. The solution of this problem does not seem to be directly applicable to the description of overlapping marginals. 
 
 Note that the task of energy minimisation can be reduced to the problem of finding the possible spectra of the one-particle reduced density matrices. This is because any quantum state can be transformed by a proper change of basis to a state, whose one-particle reduced density matrices are diagonal. The problem of finding the criteria that allow one to decide whether a given spectrum is a spectrum of some one-particle reduced density matrix has a long history. The desired criteria have a form of polygonal inequalities for the eigenvalues of the one-particle reduced density matrices. In other words, the set of solutions is a convex polytope. This is a consequence of a deep theorem about convexity properties of momentum maps, where this polytope is called the momentum 
polytope, in a theory developed simultaneously from a symplectic and algebraic point of view around 1980, in particular in works of Atiyah, Kirwan, Mumford, Ness, culminating in a general theorem on symplectic manifold by Kirwan, see \cite{K84} and the references therein. For example, in a quantum system with a fixed number of fermions, the Pauli exclusion principle \cite{Pauli} says that the spectrum of the one-particle reduced density matrix is a set of numbers between $0$ and $1$. In other words, the 
fermionic occupation numbers cannot be greater than $1$. However, these are not all the constraints. In 1970 systems of two and three fermions have been considered by Ruskai, Borland and Dennis, who have found other inequalities for the spectra \cite{BD,Ruskai70}. The first general algorithm for solving this problem has been presented by Klyachko \cite{Klyachko04,Klyachko05,AK08}. Recently, another algorithms have been presented in \cite{CDKW,VW}. Unfortunately, the computational complexity of these algorithms is still significant and they produce many redundant inequalities. The largest systems, for which the inequalities have been computed include the system of $3$ and $4$ fermions on $8$ levels \cite{Klyachko05} and the system of $3$ distinguishable particles on $4$ levels \cite{VW}. There also exists the solution for a system with an arbitrary number of qubits \cite{HSS03}.

The polytope, which is described by the polygonal inequalities for the spectra of the one-particle reduced density matrices, known as the momentum polytope in symplectic geometry, is also called the {\it spectral polytope} in our context. Spectral polytopes are also relevant for studying entanglement in pure quantum systems. This is because certain subpolytopes of the spectral polytope correspond to $SLOCC$ (Stochastic Local Operations assisted by Classical Communication) classes of entanglement. We say that two multipartite states are in the same $SLOCC$ entanglement class if there exists a local linear, invertible operation that transforms one state onto another. The spectral polytopes for different $SLOCC$ classes (in this context they are called entanglement polytopes) arise, when we allow the operations that transform the states asymptotically. Namely, the one-particle spectra corresponding to states that belong to the closure of a single $SLOCC$ entanglement class, form a subpolytope of the entire 
spectral polytope. This fact provides a necessary criterion for two chosen states to belong to the same $SLOCC$ class. For the general theory of this phenomenon with examples, see \cite{WDGC12,SOK12,MS15} and \cite{Ness} for the mathematical description.

In this paper, we introduce new lower and upper bounds for the spectral polytope in the form of polytopes contained in it, or containing it, respectively. These polytopes arise from geometric constructions around the locus of non-entangled states, viewed as a submanifold of the total space. We explain the general representation theoretic framework for a system with a compact symmetry group acting irreducibly. The manifolds of non-entangled states for this class of systems are fundamental objects in representation theory and algebraic geometry, where they are usually viewed in the projective state space, and constitute exactly the class of homogeneous projective varieties, also known as flag varieties. We focus on several physical scenarios including systems of many distinguishable particles, systems of many bosons, systems of many fermions and fermionic Fock space with a finite number of modes. In these scenarios, the loci of non-entangled states are: i) the separable states for systems of distinguishable particles, a.k.a. Segre variety; ii) spin coherent states for the angular momentum operator or, more generally, permamental states of bosons, a.k.a. Veronese variety; iii) Slater determinantal states for the fixed number of fermions, a.k.a. Gra{\ss}mann variety in its Pl\"ucker embedding iv) fermionic gaussian states in the fermionic Fock space, a.k.a. variety of pure spinors. Systems of distinguishable particles are particularly important in quantum information theory \cite{Horodecki}, as many quantum information protocols consider a situation, where 
distant parties are allowed to use arbitrary local quantum operations and send classical information. Scenarios with a fixed number of fermions are mainly considered in quantum chemistry, where the variational methods are widely used. In this context, the inequalities for the facets of the spectral polytope are called the generalised Pauli constraints and are useful in finding ground states of fermionic systems \cite{Klyachko09,SGC13,S15}. Finally, there has been a growing interest in the scenarios involving the fermionic Fock space in the context of quantum computations with noise \cite{Bravyi,MCT13,OGK14} and quantum entanglement \cite{SL14,LH14,OK13, OK14}. Despite such an interest, spectral polytopes in the context of fermionic Fock spaces have not been considered anywhere in the literature. We treat all the scenarios using the language of the representation theory, which allows us to state the results for all the scenarios simultaneously. In the same time, it provides proper geometric tools for dealing 
with the quantum marginal problem and the description of spectral polytopes.
In particular, the interpretation of the problem in terms of the representation theory of the symmetry group, allows one to use the momentum map \cite{GS}. As we explain in section \ref{sec:momentum_map}, the spectral polytope constitutes the image of the momentum map. Both bounds that we discuss in this work are given by cones whose vertex is the {\it highest weight} of the considered representation. As we explain in section \ref{sec:coh_orbit}, the highest weight can be understood as the set of spectra of one-particle reduced density matrices corresponding to a non-entangled state of a system of many-particles. Equivalently, the non-entangled states are the ground states of some noninteracting systems of many particles for hamiltonians that have non-degenerate spectra. To see this, consider a system of $L$ distinguishable particles with local hamiltonian $H=\sum_{k=1}^L H_k$. For each $H_k$ we compute its eigenstates that form the local one-particle bases $\ket{i},\ 0\leq i< n_k$, where $n_k$ is the 
dimension of the Hilbert space of the $k$-th particle. We arrange the states in each basis according to their energies, i.e. $\bra{i}H_k\ket{i}<\bra{j}H_k\ket{j}$ for $i<j$. Then, the ground state of the entire system is given by the product of the local ground states $\ket{00\dots 0}$, which is non-entangled. The construction of the cones bounding the spectral polytope uses the decomposition of the Hilbert space of the considered quantum system into spaces spanned by vectors obtained by exciting the ground state a certain number of times. By an excited state we understand a state which is obtained from the ground state by a sequence of creation operators from the Lie algebra of the symmetry group. Geometrically, the locus of states obtained by all possible sequences of $k$ excitations is known as the $k$-th osculating space to the orbit of non-entangled states. In the above example, the basis states with a single excitation are those, where one particle is in an excited state. Such states are of the form $\ket{
0\dots0i0\dots0}$, $i>0$. Similarly, the doubly excited states are of the form $\ket{0\dots0i0\dots0\dots0j0\dots0}$, $i,j>0$. We explain these concepts in Section \ref{sec:N2_poly}.

It turns out, that the information obtained by considering a relatively small number of excitations is sufficient to determine completely the spectral polytope for many quantum systems. Indeed, doubly excited states suffice in a certain ``generic case'' from a geometric point of view, the case where the locus of non-entangled states does not contain linear spaces of high dimension, i.e. the Gau{\ss}ian second fundamental form is nondegenerate. The question arises: {\it how many excitations are needed to determine the spectral polytope near the image of the ground state?} Larger physical systems do not always fall in this generic category, but doubly excited states suffice in some cases nonetheless. Nongeneric cases require new methods, and we are lead to an interesting observation concerning {\it spherical actions} of the symmetry group on the state space. In Section \ref{sec:computational_results} we list explicitly all quantum systems, whose spectral polytope 
corresponding to pure states is 
obtained form the doubly excited states. We also characterise these polytopes in terms of their vertices and facets. 

The manuscript is organised as follows. In sections \ref{sec:preliminaries} and \ref{sec:momentum_map} we describe the relevant representations and introduce the necessary notions from representation theory. Our aim is to provide a self-contained introduction for a reader, who is a non-specialist, so that the main results, which are formulated in sections \ref{sec:N2_poly} and \ref{sec:computational_results}, can be accessed more easily easier. In section \ref{sec:cone_at_hw} we introduce the mathematical background for the main results. Most of section \ref{sec:cone_at_hw} can be skipped by a reader, who is not interested in the mathematical details.

\section{Preliminaries}\label{sec:preliminaries}
We are studying the finite-dimensional systems of distinguishable and indistinguishable particles. Each scenario we introduce below concerns a compact, connected semisimple group $K$ and its complexification $G$ that act on a finite dimensional Hilbert space $\mH$, which is the space of pure (unnormalised) states of a quantum system. Let us next specify the scenarios. Firstly, we give a detailed description for the distinguishable case and the indistinguishable setting will be described via analogy. The Hilbert space for distinguishable particles is the tensor product of one-particle spaces
\[\mH_D=\CC^{N_1}\otimes\CC^{N_2}\otimes\dots\otimes\CC^{N_L},\]
The product basis of $\mH_D$ consists of tensor products of one-particle basis vectors. 
\begin{equation}\label{basis_d}
\mH_D=\SpC{\ket{i_1}\otimes\ket{i_2}\otimes\dots\otimes\ket{i_L}:\ 1\leq i_k\leq N_k},
\end{equation}
where $\langle \cdotp\rangle_{\CC}$ denotes the linear span over $\CC$. The allowed operations are local and invertible. In particular, we consider the action of local unitary (LU) operations, which is linear on $\mH_D$. A tuple of local unitary operators $U=(U_1,U_2,\dots,U_L)$, $U_k\in U(N_k)$ acts on a basis vector in the following way:
\[U\ket{i_1}\otimes \ket{i_2}\otimes\dots\otimes \ket{i_L}=U_1\ket{i_1}\otimes U_2\ket{i_2}\otimes\dots\otimes U_L\ket{i_L}.\]
Group $U(N_k)$ acts on $\CC^{N_k}$ via $N_k\times N_k$ matrices with complex entries that satisfy the condition $UU^\dagger=\bone_{N_k}$. We denote the group of $LU$ operators by $K$. Another group of local operations that we use is the group of invertible $SLOCC$ operations, which are the general linear operators, $G=GL(N_1)\times GL(N_2)\times\dots\times GL(N_L)$.

Let us next set up the notation and introduce the basic notions regarding the Lie algebras $\mkgl(N)$ and $\mku(N)$, which are the building block in all of the above cases. A fixed unitary basis $|1\rangle,...,\ket{N}$ allows us to represent $\mkgl(N)$ as represented by all $N\times N$ matrices with complex entries, whereas $\mku(N)$ is represented by antihermitian complex matrices, i.e. $X=-X^\dagger$. A natural basis for $\mk{gl}(N)$ consists of the matrices with a single nonzero entry equal to 1:
\[E_{i,j}:=\kb{i}{j}\in\mk{gl}(N),\ 1\leq i,j\leq N.\]
The hermitean conjugates of the basis vectors are simply their transpositions, $E_{i,j}^{\dagger}=E_{j,i}$, and $\mk{u}(N)$ has a basis given by
\begin{eqnarray*}
X_{i,j}:=\i(E_{i,j}+E_{j,i}),\ Y_{i,j}:=E_{i,j}-E_{j,i},\ 1\leq i<j\leq N,\\
\iota H_i , \; 1\leq i\leq N \;, 
\end{eqnarray*}
where
\begin{equation}\label{unitary_diag_op}
H_i:=\kb{i}{i},\ 1\leq i\leq N.
\end{equation}
denotes the basis of the diagonal matrices. These form a maximal abelian subalgebra consisting of semisimple elements, denoted $\mk t_N$, called a Cartan subgalgebra or a maximal torus. The rest of the basis vectors are eigenvectors for $\mk t_N$ with respect to conjugation (the adjoint action). The eigenvalues $\alpha_{i,j}$, viewed as functionals on $\mk t_N$, are called the roots of $\mk{gl}(N)$. Correspondingly, $E_{i,j}$ are called root-operators. The upper triangular matrices form a maximal solvable subalgebra of $\mk{gl}(N)$, called a Borel subalgebra. The corresponding root $\alpha_{i,j}$, $i<j$, are called positive roots, while the roots of the lower triangular matrices are called negative roots. We have indeed $-\alpha_{i,j}=\alpha_{j,i}$.

%The basis for diagonal matrices is formed by commutators
%\begin{eqnarray*}
%H_{i}:=[E_{i,i+1},E_{i,i+1}^\dagger]=\kb{i}{i}-\kb{i+1}{i+1},\ 1\leq i< N.
%\end{eqnarray*}
%Hence,
%\[\mkgl(N)=\langle \bone\rangle_\CC\oplus\langle H_N\rangle_\CC\oplus\bigoplus_{i=1}^{N} \langle %H_N\rangle_\CC\oplus\bigoplus_{i<j}\langle E_{i,j}\rangle_\CC\oplus\langle E_{i,j}^\dagger\rangle_\CC\]

Thus we have
\[\mku(N)=\bigoplus_{i=1}^{N} \langle \i H_i\rangle_\RR\oplus\bigoplus_{1\leq i< j\leq N}\langle X_{i,j},Y_{i,j}\rangle_\RR,\]
and $\mkgl(N)$ is the complexification of $\mku(N)$. Equivalently,
\[\mkgl(N)=\bigoplus_{i=1}^{N} \langle H_i\rangle_\CC\oplus\bigoplus_{1\leq i< j\leq N}\langle E_{i,j},E_{i,j}^\dagger\rangle_\CC.\]
 Both algebras are naturally equipped with a non-degenerate bilinear form, which is the Hilbert-Schmidt product:
\[(X,Y)=\tr(X^\dagger Y).\]
The Hilbert-Schmidt product is positive definite, when restricted to $\mku(N)$. Note that the introduced basis of $\mku(N)$ is orthogonal with respect to the Hilbert-Schmidt product, hence we have an isomorphism $\mku(N)\simeq \RR^N\oplus\RR^{N\choose 2}$.

In the case of distinguishable particles, the considered algebras are direct sums of the algebras of the components, i.e.
\[\mkg=\bigoplus_{k=1}^L \mkgl(N_k),\ \mkk=\bigoplus_{k=1}^L \mku(N_k).\]
An algebra element $X=(X_1,\dots,X_L)$, where $X_k\in\mkgl(N_k)$ or $X_k\in\mku(N_k)$ respectively, is represented on $\mH_D$ as matrix of the form
\begin{equation}\label{algebra_d}
X_1\otimes\bone\otimes\dots\otimes\bone+\bone\otimes X_2\otimes\bone\dots\otimes\bone+\dots+\bone\otimes\dots\otimes\bone\otimes X_L.
\end{equation}
Choosing $\{X_i\}_{i=1}^L$ to be elements of the introduced basis for the components of the respective algebras, we obtain a basis of $\mkg$ and $\mkk$.

The case of indistinguishable particles covers fermions and bosons. The Hilbert space for $L$ bosons or fermions on $N$ modes (levels) are respectively the symmetric and antisymmetric tensors from $H_D=\left(\CC^N\right)^{\otimes L}$. We denote those spaces in the following way
\[\mH_B=S^L(\CC^N),\ \mH_F=\Lambda^L(\CC^N).\]
Note that the number of fermions cannot be greater than the number of modes, $L\leq N$. The basis of $\mH_B$ consists of symmetric products of basis vectors from $\CC^N$, while the basis of $\mH_F$ consists of exterior products of basis vectors from $\CC^N$.
\begin{eqnarray}\label{basis_b}
\mH_B=\SpC{ \ket{i_1}\vee\ket{i_2}\vee\dots\vee\ket{i_L}:\ i_1\leq i_2\leq\dots\leq i_L}, \\
\mH_F=\SpC{\ket{i_1}\wedge\ket{i_2}\wedge\dots\wedge\ket{i_L}:\ i_1< i_2<\dots< i_L}. 
\label{basis_f}
\end{eqnarray}
For bosons, we alternatively denote the basis vectors by counting the mode polulations. Vector $\ket{n_1,n_2,\dots,n_N}$ corresponds to $\ket{i_1}\vee\ket{i_2}\vee\dots\vee\ket{i_L}$ such that $n_k:=\#\{l:i_l=k\}$. In both cases the considered groups are $K=U(N)$ and $G=GL(N)$, which act via the diagonal action, i.e.
\[U\ket{i_1}\vee \ket{i_2}\vee\dots\vee \ket{i_L}=U\ket{i_1}\vee U\ket{i_2}\vee\dots\vee U\ket{i_L},\]
and the same for fermions. The Lie algebra elements are represented by the following matrices
\begin{equation}\label{algebra_bf}
X\otimes\bone\otimes\dots\otimes\bone+\dots+\bone\otimes\dots\otimes\bone\otimes X,\ X\in\mkk{\rm\ or\ }X\in\mkg.
\end{equation}

\paragraph*{Fermionic Fock space} 
The last scenario that we consider is a bit different from the previous ones, as group $K$ is the real spin group $K=Spin(2N)$ and $G$ is the complex spin group, $G=Spin(2N,\CC)$. We direct the reader to \cite{MichalPhD} for a detailed description of this setting from the physical point of view. Groups $K$ and $G$ act on Majorana operators defined on the fermionic Fock space. The fermionic Fock space is the direct sum of all possible fermionic spaces in a $N$ mode system, i.e.
\[\mF=\bigoplus_{L=0}^N \Lambda^L(\CC^N),\]
where $\Lambda^0(\CC^N)=\SpC{\ket{\Omega}}\simeq \CC$ is the vacuum. Recall that a basis of the Fock space can be constructed by acting with a sequence of creation operators on the vacuum state, i.e. 
\begin{equation}\label{basis_fock}
\fl \mF=\SpC{\vac}\oplus\SpC{a_{i_1}^\dagger a_{i_2}^\dagger\dots a_{i_L}^\dagger\vac:\ 1\leq i_1<i_2<\dots<i_L\leq N,\ 1\leq L\leq N}.
\end{equation}
The elements of the above basis are the same as vectors defined in equation (\ref{basis_f}), i.e. $a_{i_1}^\dagger a_{i_2}^\dagger\dots a_{i_L}^\dagger\vac=\ket{i_1}\wedge\ket{i_2}\wedge\dots\wedge\ket{i_L}$. The creation operators satisfy the anticommutation relations $\{a_i,a_j\}=0$, $\{a_i,a_j^\dagger\}=\delta_{ij}\bone$. The $Spin(2N)$ group acts via unitary operations on the $2N$ {\it Majorana operators}, which are
\[c_{2i-1}=a_i+a_i^\dagger,\ c_{2i}=\i (a_i-a_i^\dagger),\ i\in\{1,2,\dots,N\}.\]
The Majorana operators anticommute, i.e. $\{c_i,c_j\}=2\delta_{ij}\bone$. The action reads
\[Uc_i U^\dagger=\sum_j R_{i,j}c_j,\]
where the matrix $R$ belongs to the special orthogonal group $SO(2N)$, i.e. $R^TR=\bone$, ${\rm det} R=1$. Any state from the Fock space can be also written in terms of Majorana operators, following \cite{MCT13}, as $\kpsi=\hat\gamma\vac$, where
\[\hat\gamma=\alpha_0\bone+\sum_{k=1}^N\i^k\sum_{1\leq i_1<i_2<\dots<i_{2k}\leq 2N}\alpha_{i_1,i_2,\dots,i_{2k}}c_{i_1}c_{i_2}\dots c_{i_{2k}}.\]
The action of $U\in Spin(2N)$ on $\kpsi$ is defined via the action of $U$ on $\hat\gamma$ as $U\hat\gamma U^\dagger$, which boils down to the following action on the components of $\hat\gamma$
\[(Uc_{i_1}U^\dagger)(Uc_{i_2}U^\dagger) \dots (Uc_{i_{2k}}U^\dagger).\]
There are two irreducible components of the representation of $Spin(2N)$. The first irreducible component is the subspace with an even number of fermions, $\mF_e=\bigoplus_{K=0}^{\floor{N/2}}\Lambda^{2K}(\CC^N)$ and the second irreducible component is the subspace with an odd number of fermions $\mF_o=\bigoplus_{K=0}^{\floor{N/2}}\Lambda^{2K+1}(\CC^N)$. The Lie algebra of $Spin(2N)$ is the same as the Lie algebra of $\mk{o}(2N)$ and is given by antisymmetric $2N\times 2N$ matrices with real entries. The complexified algebra $\mk{o}(2N)^\CC$ is represented on the Fock space, generated by the operators $\frac{1}{2}c_ic_j,\ 1\leq i<j\leq 2N$ giving a basis. These operators form an orthogonal basis of $\mk{o}(2N)^\CC$ with respect to the Hilbert-Schmidt norm, and the elements have squared norm
\[\frac{1}{4}\tr(c_ic_j(c_ic_j)^\dagger)=\frac{1}{4}\tr(c_ic_jc_jc_i)=\frac{1}{4}\tr\bone=\frac{1}{4}2^N.\]
The action of $\mk{o}(2N)$ on the components of $\hat\gamma$ is given by commutators
\begin{equation}\label{algebra_fock}
([c_{i_1},X])([c_{i_2},X]) \dots ([c_{i_{2k}},X]).
\end{equation}
In order to introduce the root decomposition of $\mk{o}(2N)$ and its complexification $\mk{o}^\CC(2N)$, we will use the creation and annihilation operators instead of the Majorana operators. The basis of the nondiagonal part given by root operators, which are pairs of annihilation operators of the form $a_i a_j,a_ia_j^\dagger,\ 1\leq i<j\leq N$ and their hermitian conjugates. The diagonal part of $\mk{o}(2N)$ is constructed from the shifted occupation number operators
\begin{equation}\label{ffs-diagonal}
H_i:=\frac{1}{2}(2a_i^\dagger a_i-\bone)=-\frac{\i}{2}c_{2i-1}c_{2i},\ 1\leq i\leq N.
\end{equation}
In such a basis, the complex Lie algebra is given by
\[\mk{o}(2N)^\CC=\bigoplus_{i=1}^N\spC{H_i}\oplus\bigoplus_{1\leq i<j\leq N}\spC{a_j^\dagger a_i^\dagger,a_ia_j,a_ia_j^\dagger,a_ja_i^\dagger}.\]
The nondiagonal part of $\mk{o}(2N)$ is spanned by 
\begin{eqnarray*}
A_{i,j}=\i(a_j^\dagger a_i^\dagger+a_i a_j)=\frac{\i}{2}(c_{2i}c_{2j}-c_{2i-1}c_{2j-1}),\\ B_{i,j}=a_j^\dagger a_i^\dagger-a_i a_j=-\frac{\i}{2}(c_{2i}c_{2j-1}+c_{2i-1}c_{2j}), \\
C_{i,j}=a_ia_j^\dagger-a_ja_i^\dagger=\frac{1}{2}(c_{2i-1}c_{2j-1}+c_{2i}c_{2j}), \\
D_{i,j}=\i(a_ia_j^\dagger+a_ja_i^\dagger)=\frac{1}{2}(c_{2i}c_{2j-1}+c_{2i-1}c_{2j}).
\end{eqnarray*}
Hence, 
\[\mk{o}(2N)=\bigoplus_{i=1}^N\spR{\i H_i}\oplus\bigoplus_{1\leq i<j\leq N}\spR{A_{i,j},B_{i,j},C_{i,j},D_{i,j}}.\]

\paragraph*{The common description} 
Let us formulate the above scenarios in the language of representation theory of compact and complex reductive Lie groups and Lie algebras. The results of this paper will be formulated in these terms, as they allow to treat all the scenarios simultaneously. The necessary theory for compact groups can be found for instance in the book \cite{Broe-TomDie}, while the Lie algebra aspect - in \cite{Humphreys}. In each of the scenarios we had a compact connected group $K$ and its complexification $G$ acting by an irreducible representation on a Hilbert space $\mH$, $K$ by unitary and $G$ by complex linear transformations. Although we are primarily interested in irreducible representations, we shall introduce most part of the setting without assuming irreducibility. In fact some intermediate results about reducible representations are needed for our considerations, concerning for instance the behaviour of $\mH$ with respect to subgroups of $K$, under which $\mH$ might be reducible. We will return to the notation for 
the specific cases in Section \ref{sec:computational_results}.

We consider a finite dimensional unitary representation of a compact connected Lie group $K\to U(\mH)$. The complexification of $K$ is a complex reductive Lie group with a complex (algebraic) representation $G\to GL(\mH)$. There are derived Lie algebra representations $\mk k\to\mk u(\mH)$ and $\mk g\to\mk{gl}(\mH)$. Let us fix a Cartan subgroup $T\subset K$, characterised as a maximal torus, with Lie algebra $\mkt\subset\mk k$. $T$ is a (maximal) diagonalizable subgroup in any (infinitesimally injective) representation of $K$, so $\mH$ admits a basis of $T$-eigenvectors, which are also $\mk t$-eigenvectors, called weight vectors. (For our scenarios, these are the basis vectors defined in (\ref{basis_d}), (\ref{basis_b}), (\ref{basis_f}) and (\ref{basis_fock}).) The eigenvalue is given by a function $T\to\CC^*$, which is a group homomorphism, called a character. The derived Lie algebra map is a complex linear functional $\mk t\to\CC$, called a weight, with values in $i\RR$ since the group $K$ acts unitarily. 
Identifying $\mk t$ with $i\mk t^*$ via the Hilbert-Schmidt inner product and using the Dirac notation $\vert\eta\rangle$ for weight vectors of weight $\eta$, we write
\[H_\xi\ket{\eta}=(\eta,\xi)\ket{\eta}.\]
with $\xi,\eta\in\mk{t}\simeq\RR^r$ representing respectively the acting element and the weight ($r$ is the rank of the Lie algebra). The identification $\mk{t}\simeq\RR^r$ of diagonal operators with an abstract Cartan algebra element is made via choosing an orthogonal basis of $\mkt$. The basis diagonal operators for different scenarios are the following: i) Eq.(\ref{unitary_diag_op}) for a single particle, ii) distinguishable particles -- Eq.(\ref{algebra_d}) wich each $X_k$ equal to some $H_i$ from equation (\ref{unitary_diag_op}), iii) bosons and fermions -- Eq.(\ref{algebra_bf}) with $X=H_i$ for $H_i$ from (\ref{unitary_diag_op}), iv) fermionic Fock space -- Eq.(\ref{ffs-diagonal}). With an abuse of notation, let us reenumerate the diagonal operators by $(H_i)_{i=1}^r$. Then, for $\xi=(\xi_1,\dots,\xi_r)$, we have
\[H_\xi:=\xi_1H_1+\dots\xi_r H_r.\]
The elements of $\mk t^*\cong\mk t$ arising as derivations of characters of the torus $T$ form a lattice $\Lambda$, naturally embedded as a copy of $\ZZ^r$ inside the Lie algebra $\mk t$. All weights of finite dimensional representations of $K$ are elements of $\Lambda$. We denote by $\mH^\eta$ the weight space in $\mH$ of weight $\eta$ and $Supp(\mH)\subset\Lambda$ we denote the set of weights occuring in $\mH$. Thus the action of $T$ on $\mH$ determines the weight space decomposition of the representation:
$$
\mH \cong \bigoplus\limits_{\eta\in Supp(\mH)} \mH^\eta \;.
$$
The dimension $\dim \mH^\eta$ is called the multiplicity of the weight. If $\dim \mH^\eta=1$ the weight vector is unique up to scalar; we say that the weight is multiplicity-free. We call a representation weight-multiplicity-free\footnote{We use the longer name to avoid confusion with the several existing notions of ``multiplicity-free representations'', depending on which multiplicity one refers to. For instance, in \cite{Knop-MultFreeSpa} the author calls ``multiplicity-free'' what we call here by its other popular name ``spherical''.} if all of its weight spaces are one-dimensional. Note that all our scenarios have this property. For such representations, the basis of weight vectors in $\mH$ is uniquely defined by $T$, up to independent scaling, This is not the case, for instance, with the adjoint representation where the role of $\mH$ is taken by the complex Lie algebra $\mk g$, whose weight space decomposition and basic properties we recall next.

The structure of a compact connected Lie group is encoded in its adjoint representation on the Lie algebra, its complexification $Ad:G\to GL(\mk g)$. (For a group linearly represented on a vector space $\mH$, provided the representation is injective on the Lie algebra, the adjoint action can be expressed by conjugation of the linear operators, $Ad_g(H_\xi)=gH_\xi g^{-1}$, for $g\in G$ and $\xi\in\mk g$, where we abuse notation identifying the group elements with the linear operators on $\mH$ given by the representation. The result is independent on the choice of $\mH$.) The weight space decomposition of $\mk g$ with respect to the fixed Cartan subgroup $T$ is called the root space decomposition and the nonzero weights are called roots; these were already mentioned for our examples. The weight $0$ occurs and its weight space is the complexified Cartan subalgebra $\mk t^\CC = \mk g^0$. The set of nonzero weights $\Delta=Supp(\mk g)\setminus\{0\}$ is called the root system. The root spaces $\mk g^\alpha$, $\alpha\in\Delta$ are one-
dimensional, and we pick generators $E_\alpha$, called root-vectors, or root-operators when a representation $\mH$ is given. Thus the root space decomposition has the form
\[ \mk g = \mk t^\CC \oplus (\bigoplus\limits_{\alpha\in\Delta}) \langle E_\alpha\rangle_\CC .\]
It is known from the structure theory of semisimple Lie groups that the negative ones can be chosen so that $E_\alpha^\dagger=E_{-\alpha}$, so that
\[ \mk k = \mk t \oplus (\bigoplus\limits_{\alpha\in\Delta_+} \langle E_{\alpha}-E_{-\alpha} , \iota(E_\alpha+E_{-\alpha})\rangle_\CC) \;. \]

The most basic property of root-operators is that, for any representation $\mk g \to \mk{gl}(\mH)$ and any weight $\eta$, with corresponding weight space $\mH^\eta=\SpC{\ket{\eta}}$, we have
\[ E_\alpha : \mH^\eta \to \mH^{\eta+\alpha} \quad,\quad E_\alpha\ket{\eta}\propto\ket{\eta+\alpha}.\] 
Any oriented hyperplane in $\mk t\setminus\Delta$ defines a decomposition of the root system into positive and negative roots, $\Delta=\Delta_+\cup\Delta_-$, which satisfies $\Delta_+=-\Delta_-$. We fix one such decomposition. The set of root-operators is accordingly split in two subsets, known in physics as creation and annihilation operators. Our convention is to associate positive roots to annihilation operators; in this way the ground state of a system, defined by the vanishing of the annihilation operators will correspond to a highest weight vector with the standard conventions of representation theory. The positive root-operators, together with the Cartan subalgebra span a subalgebra of $\mk g$,
\[ \mk b =  \mk t^\CC \oplus \SpC{E_\alpha ,\; \alpha\in\Delta_+} \]
which is a maximal solvable subalgebra, called a Borel subalgebra.

The Cartan-Weyl theorem, classifying the irreducible representations of reductive complex Lie algebras (and, consequently, of compact connected groups), states that in every irreducible representation $\mk g\to\mk{gl}(\mH)$, the Borel subalgebra $\mk b$ has a unique, up to scalar, eigenvector, called the highest weight vector of $\mH$, denoted here by $|\lambda\rangle$. Its weight $\lambda$ is called the highest weight and determines the irreducible representation up to isomorphism. (See e.g. \cite{Broe-TomDie},\cite{Humphreys}.)

Let us go back to our examples, starting with distinguishable particles. Taking $\xi$'s as standard basis vectors in $\RR^m$, it is easy to see that the weights for the distinguishable scenario are of the form
\begin{eqnarray*}
\left((\eta^{(1)}_1,\eta^{(1)}_2,\dots,\eta^{(1)}_{N_1}),\dots,(\eta^{(L)}_1,\eta^{(L)}_2,\dots,\eta^{(L)}_{N_L})\right),
\end{eqnarray*}
where for each $k$ we have $\eta_i^{(k)}\in\{0,1\}$ and $\#\{i:\ \eta_i^{(k)}=1\}=1$. For $S^L(\CC^N)$, we have
\[(\eta_1,\eta_2,\dots,\eta_{N}):\ \eta_i\in\{0,1,\dots,L\}{\rm\ and\ }\sum_i \eta_i=L.\]
For $\Lambda^L(\CC^N)$, we have
\[(\eta_1,\eta_2,\dots,\eta_{N}):\ \eta_i\in\{0,1\}{\rm\ and\ }\#\{i:\ \eta_i=1\}=L.\]
The positive root operators have the form
\begin{eqnarray}
E_{i_1,j_1}^{(1)}\otimes\bone\otimes\dots\otimes\bone,\dots,\bone\otimes\dots\otimes\bone\otimes E_{i_L,j_L}^{(L)},\ 1\leq i_k<j_k\leq N_k, \\
E_{i,j}\otimes\bone\otimes\dots\otimes\bone+\dots+\bone\otimes\dots\otimes\bone\otimes E_{i,j},\ 1\leq i<j\leq N, \\
a_i a_j,\ a_i a_j^\dagger\ 1\leq i<j\leq N.
\end{eqnarray}
for the distinguishable, indistinguishable and Fock space scenarios, respectively.

For the algebra $\mku(N)$, all the roots are of the form
\[(\alpha_1,\dots,\alpha_N):\ \alpha_i\in\{-1,0,1\},\ \#\{i:\ \alpha_i=1\}=\#\{i:\ \alpha_i=-1\}=1,\]
and the positive roots satisfy the condition that if $\alpha_i=1$ and $\alpha_j=-1$, then $i<j$.
The roots of $\mk{o}(2N)$ algebra read 
\begin{eqnarray*}
(\alpha_1,\dots,\alpha_N):\ \alpha_i\in\{-1,0,1\},\ \#\{i:\ \alpha_i\neq 0\}=2,
\end{eqnarray*}
i.e. signs of the nonzero components vary independently. The positive roots are such that if $\alpha_i\neq 0$ and $\alpha_j\neq 0$ for $i<j$, then $\alpha_i=1$ and $\alpha_j=\pm1$.

There are two properties of the representations corresponding to scenarios that involve fermions and distinguishable particles that make our considerations simpler. Namely, all the aforementioned representations are {\it minuscule} and {\it weight-multiplicity-free}. The latter ones were introduced above. Minuscule representations are representations, for which all weights are extreme points of $\conv{Supp(\mH)}$, cf. e.g. \cite{LM}). Minuscule representations are weight-multiplicity-free, since the extreme weights are multiplicity-free for all irreducible representations.

Let us next define a subgroup of $K$, which will play a central role in the general construction of the spectral polytope from doubly excited states.
\begin{definition}[Projective stabiliser of $\kpsi$]
The projective stabiliser of $\kpsi$ is the subgroup $K_{[\Psi]}\subset K$, which consists of elements of $K$, for which $\kpsi$ is an eigenvector, i.e.
\[K_{[\Psi]}:=\{k\in K:\ k\kpsi=c\kpsi{\rm\ for\ some\ }c\in\CC\}.\]
\end{definition}
\noindent The name {\it projective} stems from the fact that subgroup $K_{[\Psi]}$ does not change under the complex scaling of $\kpsi$. In other words, $K_{[\Psi]}$ is the stabiliser of the complex line over $\kpsi$, which is a point in the complex projective space $\PPH$.

\section{Spectral polytope via the momentum map}\label{sec:momentum_map}
In this section, we introduce a geometric description of the one-particle reduced density matrices. These matrices appear naturally when one considers the {\it momentum map}, which we denote by $\mu$. The momentum map assigns to a state $\ket{\Psi}$ an element of the Lie algebra $\mkk$ in the following way.
\begin{equation}\label{momentum_def}
\mu:\ \left(\mu(\ket{\Psi}),X\right)=\frac{1}{i}\frac{\bra{\Psi}X\ket{\Psi}}{\bk{\Psi}{\Psi}}{\rm\ for\ all\ }X\in\mkk,
\end{equation}
where $(\cdotp,\cdotp)$ denotes the Hilbert-Schmidt product. Note that the momentum map is invariant under complex scaling of $\kpsi$, i.e. $\mu(c\kpsi)=\mu(\kpsi)$, $c\in\CC-\{0\}$. Therefore, we can restrict to considering only normalised states $\bk{\Psi}{\Psi}=1$. For a normalised state, equation (\ref{momentum_def}) can be rephrased in terms of the expectation values of hermitian observables $\i X$
\[\mu:\ \tr\left(\mu(\ket{\Psi})X\right)=\EE_\Psi(\i X){\rm\ for\ all\ }X\in\mkk.\]
Therefore, $\mu(\ket{\Psi})$ is the unique operator in $i\mathfrak{k}$ whose expectation values in $i\mathfrak{k}$ coincide with those of $\kpsi$. In the distinguishable scenario, when $\i X$ are from the set of local observables (see formula (\ref{algebra_d})), the above equation becomes nothing but the definition of the one-particle reduced density matrices. In other words,
\[\mu:\ \ket{\Psi}\mapsto \left(\i \rho_1(\Psi),\i \rho_2(\Psi),\dots,\i \rho_L(\Psi)\right)\in\i\mkk,\]
where
\[\fl \rho_k(\Psi):\ \EE_\Psi(\bone\otimes\dots\otimes\bone\otimes \i X_k\otimes\bone\otimes\dots\otimes\bone)=\tr(\rho_k(\Psi)X_k){\rm\ for\ all\ }\i X_k\in\mku(N_k).\]
By choosing $\i X_k=\bone$, we see that $\tr\rho_k(\Psi)=1$. Similarly, for the scenario with indistinguishable particles we get that the momentum map assigns to a state its one-particle reduced density matrix
\[\mu:\ket{\Psi}\mapsto \i \rho_1(\Psi),\]
where
\begin{eqnarray*}
\rho_1(\Psi):\ \EE_\Psi(\i X\otimes\bone\otimes\dots\otimes\bone+\dots+\bone\otimes\dots\otimes\bone\otimes \i X)= \\ =\tr(\rho_1(\Psi)\i X){\rm\ for\ all\ }X\in\mku(N).
\end{eqnarray*}
Here, choosing $\i X=\bone$ yields $\tr\rho_1(\Psi)=L$. In the fermionic Fock space scenario, the momentum map gives the coefficients of the {\it correlation matrix} of state $\ket{\Psi}$, which is a real, antisymmetric matrix with entries given by 
\[M_{k,l}(\kpsi):=\i\frac{\bpsi c_kc_l\kpsi}{\bk{\Psi}{\Psi}}.\]
Then, using the Hilbert-Schmidt product, we have the identification
\begin{eqnarray*}
\mu(\kpsi)=\frac{1}{2^{N-2}}\Bigg{(}\sum_{i=1}^N M_{2i-1,2i}\i H_i+\sum_{1\leq i<j\leq N}2\left(M_{2i,2j}-M_{2i-1,2j-1}\right)X_{i,j}+\\ +\sum_{1\leq i<j\leq N}2\left(M_{2i,2j-1}-M_{2i-1,2j}\right)Y_{i,j}\Bigg{)}.
\end{eqnarray*}
Therefore, the problem of determining the set of possible one-particle reduced density matrices for systems of distinguishable or indistinguishable particles and the problem of finding the set of possible correlation matrices for the fermionic Fock space, is equivalent to the problem of describing the image of the momentum map defined in (\ref{momentum_def}) stemming from a representation of a compact connected group $K$. So far, we haven't used the most important property of the momentum map, namely, its $K$-equivariance. Indeed,
\[ \mu : \mH\setminus \{0\} \to \mk{k} \]
is compatible with the $K$-actions given by the representation on $\mH$ and the adjoint action on $\mk{k}$, respectively, in the sense that
\[\mu(U\ket{\Psi})=U^\dagger\mu(\ket{\Psi})U,\ U\in K.\]
This $K$-conjugation in $\mk k$ preserves momentum image. This property allows us to parametrise the image of $\mu$ by its intersection with the subspace $\mk{t}$ of the diagonal matrices, since any matrix from the Lie algebra $\mkk$ is $K$-conjugate to an element of $\mk{t}$. In the scenarios with a fixed number of particles, the image of $\mu$ is parametrised by the spectra of the one-particle reduced density matrices. We call these spectra the {\it local spectra}. For the Fock space scenario, recall that any real antisymmetric matrix $M$ can be brought by the orthogonal transformations to a form, where the only nonzero entries are the occupation numbers $M_{2i-1,2i}$. This is equivalent to diagonalising the matrix $\mu(\kpsi)$. In order to make the parametrisation unambiguous, we order the eigenvalues decreasingly. The set of diagonal matrices with decreasingly ordered local spectra is called the {\it positive Weyl chamber}, which we denote by $\mkt_+$.
\begin{eqnarray*}\fl
\mkt_+=\{(\eta_1,\eta_2,\dots,\eta_N):\ \eta_1\geq\eta_2\geq\dots\geq\eta_N\geq0\}{\rm\ \ for\ \ }\mkk=\mku(N), \\ \fl
\mkt_+=\mkt_+^{(1)}\oplus\mkt_+^{(2)}\oplus\dots\oplus\mkt_+^{(L)}{\rm\ \ for\ \ }\mkk=\mku(N_1)\oplus\dots\oplus\mku(N_L),\\ \fl
\mkt_+=\{(\eta_1,\eta_2,\dots,\eta_N):\ \eta_1\geq\eta_2\geq\dots\geq\eta_N,\ \eta_{N-1}+\eta_N\geq0\} {\rm\ \ for\ \ }\mkk=\mk{o}(2N).
\end{eqnarray*}
In this notation, the image of the momentum map can be described as
\[\mu(\mH)=K^\dagger(\mu(\mH)\cap\mkt_+)K.\]
The celebrated convexity theorem, which has been first stated for abelian groups by Atiyah \cite{atiyah}, and later developed by many others \cite{K84,Sjamaar,Brion,BS,Smirnov,HH96} states that (in a setting much more general than the one introduced in this paper) the intersection of the momentum image with the Weyl chamber is a convex polytope, called the momentum polytope. We denote this polytope by 
\[\poly = \mu(\mH)\cap\mkt_+.\]
and refer to it as the {\it spectral polytope} according to the physical interpretation explained in Section \ref{sec:introduction}. Elements in the spectral polytope are one-body quantum marginals of pure states that are diagonal with ordered eigenvalues.

Let us next point out a key technical ingredient of our method. It concerns a distinguished class of subspaces of $\mH$, whose momentum images are easy to find, and are in fact convex polytopes spanned by sets of weights. These subspaces are spanned by weight vectors, whose corresponding weights have the property of {\it root-distinctness} defined below. The basic idea is due to Wildberger, \cite{Wildberger}, and has been developed further by Sjamaar, \cite{Sjamaar}.

\begin{lemma}\cite{Sjamaar,Wildberger}
The momentum image of a state $\ket{\Psi}$ lies in the Cartan subalgebra $\mk t$ if and only if $\bra{\Psi}E_\alpha\ket{\Psi}=0$ for all roots $\alpha\in\Delta$.
\end{lemma}

\begin{proof}
Recall that operators $X_\alpha=\i(E_\alpha+E_{-\alpha})$ and $Y_\alpha=E_\alpha-E_{-\alpha}$ form an orthogonal basis of the non-diagonal part of $\mkk$. Hence, $\mu(\kpsi)$ is diagonal if and only if $(\mu(\kpsi),X_\alpha)=(\mu(\kpsi),Y_\alpha)=0$ for all $\alpha$. This is equivalent to $(\mu(\kpsi),E_\alpha)=(\mu(\kpsi),E_{-\alpha})=0$ for all $\alpha$. By formula (\ref{momentum_def}) we obtain the claim.
\end{proof}

\begin{definition}[A root-distinct subset of weights]
A subset of weights $\Lambda'\subset\Lambda$ is called root-distinct, if for every pair $\{\eta,\nu\}\subset\Lambda'$ we have 
\[\eta-\nu\neq\alpha{\rm\ for\ all\ roots\ }\alpha\in\Delta.\]
\end{definition}

\begin{definition}[The support of a state]
The support of a state $\Sigma(\Psi)\subset\Lambda$ is the set of weights such that the corresponding weight vectors occur with nonzero coefficients in the decomposition of $\ket{\Psi}$. In other words,
\[\Sigma(\Psi):=\{\eta\in\Lambda:\ \bk{\eta}{\Psi}\neq 0\}.\]
\end{definition}

\begin{lemma}\label{lemma:torus_poly}{\cite{Atiyah}}
Consider $\mH$ as a representation of the maximal torus $T\subset K$. The torus momentum map $\mu_T$ is the $K$-momentum map composed with the projection on $\mkt$, i.e. $\mu_T={\rm pr}_\mkt\circ\mu$. In other words,
\[(\mu_T(\ket{\Psi}),\i H)=(\mu(\ket{\Psi}),\i H),\ \i H\in\mkt\subset\mkk.\]
The image of $\mu_T$ is the convex hull of weights with 
\begin{equation}\label{mu-diagonal}
\mu_T\left(\sum_{\eta\in Supp(\mH)}a_\eta\ket{\eta}\right)=\sum_{\eta\in Supp(\mH)}|a_\eta|^2\eta.
\end{equation}
\end{lemma}
\begin{proof}
Decompose $\ket{\Psi}$ with respect to the weight spaces
\[\ket{\Psi}=\sum_{\eta\in Supp(\mH)}a_\eta\ket{\eta}.\]
Equality $\mu_T={\rm pr}_\mkt\circ\mu$ is just a consequence of the definition of the momentum map (equation \ref{momentum_def}). For the elements of the Cartan algebra, we have for normalised $\kpsi$
\[(\mu(\kpsi),\i H_\xi)=\sum_{\eta,\nu\in\Sigma(\Psi)}\overline a_\nu a_\eta\bra{\nu}H_\xi\ket{\eta}=\sum_{\eta\in\Sigma(\Psi)}|a_\eta|^2 (\eta,\xi).\]
Hence, by the non-degeneracy of the Killing form $(\cdotp,\cdotp)$ we obtain (\ref{mu-diagonal}).
\end{proof}

\begin{lemma}\label{lemma:rd}
For a state, whose support is root-distinct, we have $\mu(\ket{\Psi})\in\mkt$. In particular, the weight vectors are mapped by $\mu$ to their weights.
\end{lemma}
\begin{proof}
Firstly, note that if $\Sigma(\kpsi)$ is root-distinct, then $\Sigma(E_\alpha(\kpsi))\cap\Sigma(\kpsi)=\emptyset$ for all $\alpha$. This in turn means that $\bra{\Psi}E_\alpha\ket{\Psi}=0$ for all $\alpha$, hence $\mu(\kpsi)$ is diagonal. For the diagonal part, we have formula (\ref{mu-diagonal}).
\end{proof}

\noindent In systems of distinguishable qudits, root-distinctness of set $\Lambda'\subset\Lambda$ means that for every pair $\mu(\ket{i_1,\dots,i_L}),\ \mu(\ket{j_1,\dots,j_L})\in\Lambda'$, we have $\#\{k: i_k\neq j_k\}>1$. For fermions, the pairs must be of the form $\mu(\ket{i_1}\wedge\ket{i_2}\wedge\dots\wedge \ket{i_L}),\ \mu(\ket{j_1}\wedge\ket{j_2}\wedge\dots\wedge \ket{j_L})\in\Lambda'$, with the condition $\#\{k: i_k\neq j_k\}>1$. The root-distinct sets of a maximal size have proven to be useful in finding the spectral polytopes for small fermionic systems, namely for $\Lambda^3\left(\CC^6\right)$ and $\Lambda^3\left(\CC^7\right)$ \cite{BD}. For $\Lambda^3\left(\CC^6\right)$ the weight vectors corresponding to weights from the maximal root-distinct sets are $\ket{1}\wedge\ket{2}\wedge\ket{3}$, $\ket{1}\wedge\ket{4}\wedge\ket{5}$, $\ket{2}\wedge\ket{4}\wedge\ket{6}$, $\ket{3}\wedge\ket{5}\wedge\ket{6}$. For $\Lambda^3\left(\CC^7\right)$
the weight vectors corresponding to weights from the maximal root-
distinct sets are the same as for $\Lambda^3\left(\CC^6\right)$ plus the weight vectors $\ket{1}\wedge\ket{6}\wedge\ket{7}$, $\ket{2}\wedge\ket{5}\wedge\ket{7}$ and $\ket{3}\wedge\ket{4}\wedge\ket{7}$. Using Lemma \ref{lemma:rd}, one can easily compute the part of the spectral polytope stemming from the maximal root-distinct sets of weights. The authors of \cite{BD} have shown that the local spectra of states, whose support is a maximal root-distinct set of weights, yield the entire spectral polytope for $\Lambda^3\left(\CC^6\right)$ and $\Lambda^3\left(\CC^7\right)$. To prove that, they used the knowledge of the universal space for the action of the local unitary group on the respective Hilbert spaces, which gives a canonical form of a quantum state up to the action of local unitaries.
Recently, the universal $LU$-spaces have been described for any system of three fermions \cite{universal-fermions}. However, it is difficult to apply the approach from \cite{BD} for systems larger than $\Lambda^3\left(\CC^7\right)$, as the states from the universal space have non-diagonal 1PRDMs and the number of equations one has to take into account to make the 1PRDMs diagonal, grows fast with the number of modes.

Let us next briefly review a result regarding the structure of the spectral polytope \cite{Sjamaar}. Namely, the spectral polytope can be written as an intersection of {\it local cones}
\begin{equation*}
\poly=\bigcap_{\ket{\Psi}\in\mu^{-1}(\mkt_+)}\mC(Y_\Psi),
\end{equation*}
where a local cone $\mC(Y_\Psi)$ is the polyhedral cone at vertex $\mu(\kpsi)$, spanned by the momentum polytope stemming from the action of the centraliser of $\mu(\ket{\Psi})$ on space $Y_\Psi$. Space $Y_\Psi$ is defined via the symplectic slice at $\kpsi$ (see Theorem 6.3 and Definition 6.4 in  \cite{Sjamaar}). We will not go into the details of the construction of symplectic slices, as in this paper we only consider a specific local cone, for which the description of $Y_\Psi$ is much simpler. We only mention, that the above intersection of cones is locally finite, which means that any point $\xi\in\poly$ belongs to a finite number of local polytopes.

Another important property of local cones is that if we choose $\mu(\Psi)$ to be a vertex of $\poly$, then the cone at vertex $\mu(\Psi)$ spanned by $\poly$ is precisely the local cone $\mC(Y_\Psi)$. In other words, such a local cone describes the spectral polytope around the chosen vertex \cite{Sjamaar}. In particular, this means that 
\begin{equation}\label{hw_cone_containing_poly}
\poly\subset\mC(Y_\Psi)\cap\mkt_+
\end{equation}
for $\kpsi$ being mapped to a vertex of $\poly$. The vertices of the spectral polytope $\mu(\mH)\cap\mk{t}_+$ fall into two types: ones inherited from the full momentum image and ones resulting from the intersection with the Weyl chamber. To this end there is the following theorem.

\begin{theorem}\label{thm:vertices}\cite{Sjamaar}
If $\poly$ is of full dimension and $v$ is a vertex of $\poly$ that lies in the interior of the positive Weyl chamber, then the projective stabiliser of any $\kpsi\in\mu^{-1}(v)$ is equal the maximal torus. In particular, $\kpsi$ is a weight vector.
\end{theorem}
\noindent Polytope $\poly$ is of full dimension iff its dimension is equal to the dimension of $\mkt$ minus the constraints for the normalisation of the image of $\mu$, see Section \ref{sec:momentum_map}. This is true iff the derivative $d\mu$ at a generic state is surjective or, equivalently, the stabiliser of a generic state is discrete.

By Theorem \ref{thm:vertices}, vertices of the spectral polytope are either weights, or belong to the boundary of the positive Weyl chamber. The situation simplifies, when the representation of $K$ on $\mH$ is irreducible, as is the case with all our scenarios. Then the non-entangled states form a single projective $K$-orbit, and $\mc P(\mH)$ has a distinguished vertex - the highest weight (see Fig.\ref{fig:local-cone}). The local cone at this vertex is the subject of the next section. It can be studied using the following lemma, since its hypothesis is satisfied by the ground state, as we shall show below.
\begin{figure}[h]
\centering
\includegraphics[width=.6\textwidth]{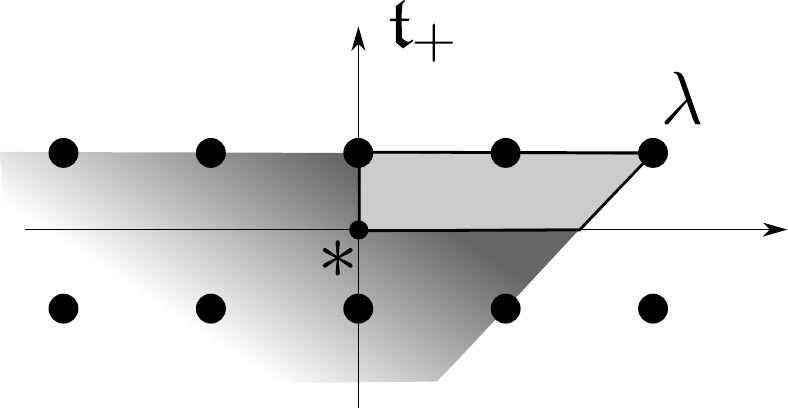}
\caption{The spectral polytope (light grey) for the representation of $SU(2)\times SU(2)$ on $\mH=\CC^2\otimes S^4\left(\CC^2\right)$, which is the system that consists of one qubit and a system of four bosons occupying two modes. In the mode population notation, the highest weight vector is $\ket{\lambda}=\ket{0}\otimes\ket{4,0}$. The entire gradient-grey shaded area is the local cone at the highest weight.}
\label{fig:local-cone}
\end{figure}

\begin{lemma}\label{lemma:Sja:locconelambda}\cite{Sjamaar}
If the centraliser of $\mu(\kpsi)$ is equal to the projective stabiliser of $\kpsi$, then the symplectic slice is isomorphic to the normal space at $\kpsi$, i.e.
\[Y_\Psi=N_\Psi:=T_{\ket{\Psi}}\mH/T_{\ket{\Psi}}(K\ket{\Psi}).\]
As a consequence, the local cone at $\kpsi$ is the cone at $\mu(\kpsi)$ spanned by the polytope $\mP(N_\Psi)$, where $N_\Psi$ is regarded as a (possibly reducible) representation of $K_{\mu\left(\kpsi\right)}$.
\end{lemma}

\section{The local cone at the ground state and osculating spaces}\label{sec:cone_at_hw}

A weight vector, which is annihilated by all positive root operators is called a {\it highest weight vector}. When the representation is irreducible, the highest weight vector is unique up to scalar. It will be denoted by \ket{\lambda}, with $\lambda$ being the highest weight. By the Cartan-Weyl theorem the irreducible representations of compact connected groups (or equivalently their complexifications) are classified by their highest weights. 

As we explained in the Section \ref{sec:introduction}, the highest weight of a representation can be viewed as the ground state of a hamiltonian with no interactions between the particles, whose spectrum is non-degenerate. Indeed, the choice of a specific hamiltonian means fixing a maximal torus $T$. This is done by taking the weight vectors to be the eigenvectors of such a hamiltonian. Then, the highest weight vector can be chosen to be the eigenvector that minimises the energy.

In all our scenarios the representations are irreducible, and the highest weight vectors and the corresponding highest weights are:
\begin{eqnarray*}
\ket{\lambda}=\ket{1}\otimes\ket{1}\otimes\dots\otimes\ket{1},\ \lambda=\left((1,0,\dots,0),\dots,((1,0,\dots,0))\right){\rm\ for\ }\mH_D, \\
\ket{\lambda}=\ket{1}\vee\ket{1}\vee\dots\vee\ket{1},\ \lambda=(L,0,0,\dots,0) {\rm\ for\ }S^L(\CC^N),\\
\ket{\lambda}=\ket{1}\wedge\ket{2}\wedge\dots\wedge\ket{L},\ \lambda=(1,1,\dots,1,0,\dots,0){\rm\ for\ }\Lambda^L(\CC^N), \\
\ket{\lambda}=\vac,\ \lambda=\left(\frac{1}{2},\frac{1}{2},\dots,\frac{1}{2}\right){\rm\ for\ }\mc{F}_e,\\
\ket{\lambda}=a_N^\dagger\vac,\ \lambda=\left(\frac{1}{2},\frac{1}{2},\dots,\frac{1}{2},-\frac{1}{2}\right){\rm\ for\ }\mc{F}_o.
\end{eqnarray*}

All $T$-weights of the representation $\mH$ are obtained by adding sequences of negative roots to the highest weight:
\[ Supp(\mH)\subset \lambda + Cone_{\ZZ_{\geq 0}}(\Delta_-) \;. \]

Let us remark that, in the scenarios not involving bosons, there is only one weight that belongs to $\mkt_+$, namely the highest weight $\lambda$. Note that the local cone at the highest weight intersected with $\mkt_+$ is in general larger than the spectral polytope (Fig.\ref{cone_scheme}).

\begin{figure}[h]
\centering
\includegraphics[width=.5\textwidth]{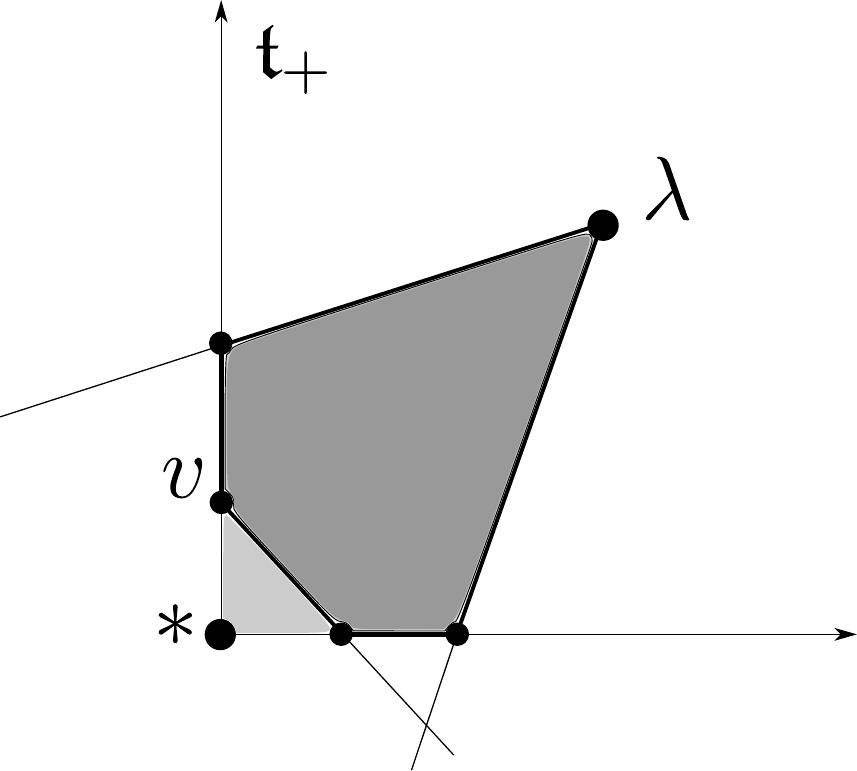}
\caption{A schematic picture showing the relation between the entire spectral polytope and the local cone at $\lambda$. The entire spectral polytope is the dark grey area. The positive Weyl chamber is the area between the arrows. The whole grey area is the cone at $\lambda$ intersected with the positive Weyl chamber. In order to obtain the spectral polytope, we have to cut a part of the grey area with the cone at vertex $v$, which lies on the boundary of $\mkt_+$.}
\label{cone_scheme}
\end{figure}

Since the choice of a particular vector along the eigenline is irrelevant, it is suitable to consider the projective space $\PP(\mH)$ and denote by $[\Psi]$ the point corresponding to the line spanned by a nonzero element $\kpsi\in\mH$. The projective $K$-orbit of the highest weight line has the form $K[\lambda]\cong K/K_{[\lambda]}$, given by the orbit-stabilizer theorem, where $K_{[\lambda]}$ is the projective stabiliser of $\ket{\lambda}$. The manifolds obtained this way share many interesting properties, they are known as flag varieties, and constitute the class of all simply connected homogeneous compact K\"ahler manifolds, when $K$ varies over all compact semisimple groups. We shall not discuss them at large but focus on some properties related to the projective embedding $K[\lambda]\subset\PP(\mH)$ and $K|\lambda\rangle\subset \mH$.

\begin{lemma}\cite{GS}
The projective stabiliser $K_{[\lambda]}$ of $[\lambda]\in\PP(\mH)$ equals the centraliser of the highest weight $\lambda\in\mk t_+$, i.e. $K_{[\lambda]}=K_{\lambda}$. Consequently, the projective orbit $K[\lambda]\in\PP(\mH)$ is isomorphic to the (co)adjoint orbit $K\lambda\subset\mk k$.
\end{lemma}

This enables us to apply Sjamaar's lemma \ref{lemma:Sja:locconelambda} in the computation of the local cone. The lemma involves the tangent and normal space to the orbit space, which are discussed in the next section. 

\subsection{The orbit of non-entangled states}\label{sec:coh_orbit}
In this section we review some key facts concerning the non-entangled states. As the main result of this section, we construct a basis of the tangent space at the ground state to the set of non-entangled states - formula (\ref{tangent-space}). 

For systems of distinguishable particles, a state is non-entangled if and only if it is separable, i.e. can be written as a simple tensor 
\begin{equation}\label{separable}
\kpsi=\ket{\phi_1}\otimes\ket{\phi_2}\otimes\dots\otimes\ket{\phi_L},
\end{equation}
where $\ket{\phi_i}$ is a vector from $\CC^{N_i}$. Any state, which cannot be written in the form (\ref{separable}) is entangled. Note that any separable state can be obtained from the highest weight vector by a $SLOCC$ operation, i.e.  $\ket{\phi_1}\otimes\ket{\phi_2}\otimes\dots\otimes\ket{\phi_L}=g_1\ket{0}\otimes g_2\ket{0}\otimes\dots\otimes g_L\ket{0}$ for some $g_1,\dots,g_L$. In other words, the set of separable states is the $G$-orbit through the highest weight vector. Extending this approach to other scenarios, we give a general geometric definition of the set of non-entangled stated.

\begin{definition}
The set of non-entangled states is the $G$-orbit $G\vert\lambda\rangle$ through the highest weight vector.
\end{definition}

\noindent For specific scenarios, we have the following forms of non-entangled states.
\begin{itemize}
\item System of $L$ fermions in $N$ modes: $\mH_F=\Lambda^L\left(\CC^N\right)$, the non-entangled states are Slater determinantal states
\[\kpsi=\ket{\phi_1}\wedge\ket{\phi_2}\wedge\dots\wedge\ket{\phi_L},\ \ket{\phi_i}\in\CC^N.\]
The set of non-entangled states in the projective space $\PP(\mH_F)$ is also the Pl\"{u}cker embedding of the Grassmannian $Gr(L,N)$.

\item System of $L$ bosons in $N$ modes: $\mH_B=S^L\left(\CC^N\right)$, the non-entangled states correspond to permanental states
\[\kpsi=\ket{\phi_1}\vee\ket{\phi_2}\vee\dots\vee\ket{\phi_L},\ \ket{\phi_i}\in\CC^N.\]
The set of non-entangled states in the projective space $\PP(\mH_B)$ is also the Veronese embedding of $\PP(\CC^N)$. In the case, when $N=2$, the set of non-entangled states is also called the set of spin-coherent states. Such a setting corresponds to one quantum spin with total angular momentum $j=L/2$.

\item Fermionic Fock spaces on $N$ modes: $\mF_e$ and $\mF_o$. The set of non-entangled states is the set of fermionic pure Gaussian states. In the mathematical literature these are called pure spinors. Their general form is \cite{chevalley}
\[\kpsi=e^A\tilde a_1^\dagger\tilde a_2^\dagger\dots\tilde a_k^\dagger\vac,\ k\leq N,\]
where $\tilde a_i^\dagger=e^{-B}a_i^\dagger e^B$. Matrices $B=\sum_{1\leq i<j\leq 2N}\frac{1}{2}B_{i,j}c_ic_j$ and $A=\sum_{1\leq i<j\leq 2N}\frac{1}{2}A_{i,j}c_ic_j$ are elements of group $Spin(2N)$.
\end{itemize}

\begin{remark}
Viewed in $\mH$, the set of non-entangled states consists of many $K$-orbits. In fact, $G\vert\lambda\rangle$ intersects each sphere in $\mH$ centered at $0$ at a single $K$-orbit and, as stated in the next lemma, the projective orbits are equal: $G[\lambda]=K[\lambda]\subset\PP(\mH)$. Since he momentum map $\mu$ is $K$-equivariant, and its definition includes a normalisation, we have 
\[\mu(G\vert\lambda\rangle) = \mu(K\vert\lambda\rangle) = K\lambda \quad,\quad K\lambda\cap\mk t_+=\{\lambda\} \;.\] 
In other words, all the non-entangled states have the same spectra of the one-particle reduced density matrices, given by the highest weight $\lambda$. 
\end{remark}

Recall that the tangent space to an orbit of a Lie group can be naturally identified with the quotient of the Lie algebra by the subalgebra of the stabiliser at the chosen point. Since we are considering the projective stabiliser of a weight vector, we have $\mk t\subset\mk{k}_\lambda$. Hence $\mk{k}_\lambda$ is preserved by the adjoint $\mk{t}$-action on $\mk{k}$ and so is its orthogonal complement which we denote by $\mk{m}$, so that we get
\[ \mk{k} = \mk{k}_\lambda \oplus \mk{m} \;\;,\;\; T_{[\lambda]}K[\lambda] \cong \mk{k}/\mk{k}_\lambda \cong \mk{m} \;.\]
It will be useful to describe the tangent space in terms of roots and to this end we consider the action of the complexified Lie group $G=K^\CC$ and its Lie algebra $\mk g$, where the root-operators belong.

Let $\i H_\lambda\in\mk{t}$ be the operator corresponding to $\lambda$. We consider its adjoint action on $\mk g$, $ad_{\i H_\lambda}(A)=[\i H_\lambda,A]$, and denote by $\mk g^\lambda_0$ its zero eigenspace (which is the centraliser of $H_\lambda$) and by $\mk{g}^\lambda_+,\mk{g}^\lambda_-$ the sum of the eigenspaces of $\iota H_\lambda$ with positive and negative eigenvalues, respectively. Let $\Delta=\Delta_0^\lambda\cup\Delta^\lambda_+\cup\Delta^\lambda_-$ be the corresponding partition of the root system. The eigenvectors of $ad_{\i H_\lambda}$ are the root operators, as $ad_{\i H_\lambda}(E_\alpha)=(\lambda,\alpha)E_\alpha$. Therefore, we have 

\[\fl\Delta_0^\lambda=\{\alpha\in\Delta:\ (\alpha,\lambda)=0\},\ \Delta_+^\lambda=\{\alpha\in\Delta:\ (\alpha,\lambda)>0\},\ \Delta_-^\lambda=\{\alpha\in\Delta:\ (\alpha,\lambda)<0\}.\]

%It is easy to see that all there $\mk{g}^\lambda_+,\mk{g}^\lambda_0,\mk{g}^\lambda_-$ are subalgebras of $\mk g$ and furthermore $\mk g^\lambda_\pm$ are preserved by $\mk g^\lambda_0$, so that $\mk g^\lambda_0\oplus\mk{g}^\lambda_\pm$ is also a subalgebra.

\begin{lemma}\label{Lemma CohOrbit}\cite{KS82}
Let $M=K[\lambda]\subset\PP(\mH)$ be the projective $K$-orbit through the highest weight vector of an irreducible representation. Then

(i) $M$ is a complex algebraic subvariety of $\PP(\mH)$, preserved by the complex group $G$. It is the unique compact projective orbit of the complex group and the unique complex projective orbit of the compact group $K$.

(ii) The stabiliser $G_{[\lambda]}$ is a subgroup of $G$, containing the Borel subgroup $B$, so-called parabolic subgroup, so that $M$ can be also written as $G/G_{[\lambda]}$. The Lie algebra of $G_{[\lambda]}$ is the sum of the eigenspaces of $\iota H_\lambda$ with nonnegative eigenvalues:
\[\mk g_{[\lambda]} = \mk{g}^\lambda_+ \oplus \mk{g}^\lambda_0 \;.\]
Furthermore, $\mk{k}_\lambda^\CC=\mk{g}^\lambda_0$. This complex Lie algebra is called the reductive Levi component of the parabolic subgroup, while $\mk{g}^\lambda_+$ is the nilpotent radical.

%(iii) The tangent space $T_{[\lambda]}M$ is isomorphic to the space $\mk g^\lambda_-$ (which is incidentally a Lie algebra) whose roots are all negative roots non-orthogonal to $\lambda$.
\end{lemma}

%Note that, the the vector space $\mH$, the compact orbit $K|\Psi\rangle$ is the intersection of the complex orbit $G|\lambda\rangle$ with the sphere through $|\lambda\rangle$. As long as $\lambda$ is nonzero, the complex torus $T^\CC$ acts transitively on the punctured line $\CC^*|\lambda\rangle$, whence $\CC|\lambda\rangle\subset T_{|\lambda\rangle}G|\lambda\rangle$. However, this complex line is not in the tangent space to the $K$-orbit $T_{|\lambda\rangle}K|\lambda\rangle$.
\noindent As a consequence, we can write the tangent space to $M$ in the following simple basis.
\begin{lemma}
Let $M = G|\lambda\rangle$ denote the complex orbit of the highest weight vector in $\mH$. Then its tangent space is given by
\begin{equation}\label{tangent-space}
T_{\ket{\lambda}}G\ket{\lambda}=\SpC{\ket{\lambda}}\oplus\SpC{E_\alpha\ket{\lambda},\ \alpha\in\Delta_-^\lambda}.
\end{equation}
The weights of $T_{|\lambda\rangle}G|\lambda\rangle$ are $\lambda$ and the set $\lambda+\Delta^\lambda_-$.
\end{lemma}
\begin{proof}
Note first that the line through $|\lambda\rangle$, i.e. the complex span $\CC\vert\lambda\rangle\subset\mH$ belongs to $T_{\ket{\lambda}}G\ket{\lambda}$. It is in fact (except for the point $0$) contained in the orbit $G|\lambda\rangle$, since $\lambda$ defines a nontrivial character of $T$ and $T|\lambda\rangle=\CC^*|\lambda\rangle$. Thus the tangent space to $G|\lambda\rangle$ splits as the sum of $\CC|\lambda\rangle$ and the tangent space to the projective orbit $G[\lambda]$.  Using the decomposition $\mk g=\mk g^\lambda_+\oplus\mk g^\lambda_0\oplus\mk g^\lambda_-$ and part (ii) of Lemma \ref{Lemma CohOrbit}, we obtain that the tangent space to $G[\lambda]$ is isomorphic to $\mk g/\mk g_{[\lambda]}=\mk g^\lambda_-$, which is in turn the span of $E_\alpha$ for $\alpha\in \Delta^\lambda_-$. This space is embedded in $\mH$ via the action at $|\lambda\rangle$, which yields the desired formula. The statement for the weights follows from the fact that any root-operator sends weight spaces to weight spaces, and 
$E_\alpha|\lambda\rangle$ has weight $\lambda+\alpha$. 
\end{proof}

\subsection{The normal space, osculating spaces and $j$-excitation spaces}

The entire vector space $\mH$ can be obtained from the highest weight vector $|\lambda\rangle$ by applying finite sequences of root operators $E_{\alpha}, \alpha\in\Delta_-$ and taking linear combinations. In the representation theoretic language, this amounts to the action of the universal enveloping algebra. We shall avoid the general terminology and confine ourselves to a concrete construction.

Consider the following sequence of subspaces of $\mH$, called the osculating spaces to the orbit $G\klambda$, obtained by acting on $\klambda$ by sequences of a given length: 
\[ O^j = O_{\klambda}^jG\klambda = \SpC{ X_1X_2...X_j \klambda \;:\; X_i\in\mk g} \;,\]
for any fixed $j\in\ZZ_{\geq 0}$. We also set $O_{\klambda}^0=\CC \klambda$. Geometrically, the $j$-th osculating space is spanned by the derivatives, up to order $j$, of (complex or real) analytic curves $t\mapsto |\Psi_t\rangle$ lying on $G\klambda$ and passing through $\klambda$.

\begin{lemma}
Denote for simplicity $O^j=O_{\klambda}^jG\klambda$. The following hold

(i) To obtain $O^j$ it is sufficient to apply only $j$-sequences of elements $E_\alpha$, with $\alpha\in\Delta_-$. The osculating spaces at $\klambda$ are preserved by $K_\lambda$ and form an exhaustion of $\mH$ by nested $K_\lambda$-subrepresentations:
\[ 0 \subset O^0 \subset O^1 \subset O^2 \subset ... \subset O^{j_{max}} = \mH \;, \]
with $O^0=\CC\klambda$ and $O^1=T_{\klambda} G\klambda$. 

(ii) Let $N_j$ denote the orthogonal complement of $O^{j-1}$ in $O^j$ for $j\geq 1$ and $N_0=O^0=\CC\klambda$. Then $\mH$ decomposes as a direct sum of $K_\lambda$-subrepresentations:
\[ \mH = N_0 \oplus N_1 \oplus N_2 \oplus ... \oplus N_{j_{max}}.\]

(iii) The normal space to $G\klambda$ at the point $\vert\lambda\rangle$ is given, upon identification with the orthogonal complement of the tangent space, by
\[ N_\lambda = \bigoplus\limits_{j\geq 2} N_j \;.\] 

\end{lemma}

\begin{proof}
The $K_\lambda$-invariance of $O^j$ follows by induction, the base case being the tangent space to the orbit. The fact that it suffices to employ only negative root vectors follows from the Poincar\'e-Birkhoff-Witt theorem (cf. \cite{Humphreys}) applied to the bases of $\mk g$ given by the root vectors and arbitrary elements form the torus. Indeed, using the commutation relations, any operator resulting from a $j$-sequence $X_1...X_j$ can be written as a sum of operators of the form $E_{\alpha_1}...E_{\alpha_p} H_{\beta_1} ... H_{\beta_q} E_{\gamma_{1}} ... E_{\gamma_r}$, with $p+q+r\leq j$ and $\alpha_i\in\Delta_-$, $\beta_i\in\Delta$, $\gamma_i\in\Delta_+$. Since the highest weight vector is fixed by the Borel subalgebra $\mk b$ and annihilated by the positive root vectors, we see that $E_{\alpha_1}...E_{\alpha_p} H_{\beta_1} ... H_{\beta_q} E_{\gamma_{1}} ... E_{\gamma_r}\klambda$ vanishes whenever $r\ne 0$, while, whenever $r=0$, it is proportional to $E_{\alpha_1}...E_{\alpha_p}\klambda$ and belongs to 
the weight space of 
the 
weight $\lambda+\alpha_1+...+\alpha_p$. Hence, Space $O^j$ is spanned by vectors $E_{\alpha_1}...E_{\alpha_p}\klambda$, $\alpha_1,...,\alpha_p\in\Delta_-$, $p\leq j$. These observations imply part (i). Part (ii) follows directly from part (i) and the definitions. Namely, if spaces $O^j$ are $K_\lambda$ invariant, then $N_j=O^j/O^{j-1}$ are also $K_\lambda$ invariant. By definition we have $N_i\cap N_j=\emptyset{\rm\ for\ }i\neq j$. For part (iii) note that $N_0\oplus N_1\cong O^1=T_{\klambda} G\klambda$. Hence, the complement is the normal space.
\end{proof}

\begin{definition}\label{def:N_j}
The spaces $N_j$ are called the $j$-excitation spaces of $\mH$, relative to the ground state $\klambda$.
\end{definition}

In the language of representation theory, and projective geometry, $N_j$ is known as the $j$-th normal space to the orbit. $N_2$ is the image of the Gau{\ss}ian second fundamental form of the embedding, and analogously $N_j$ is the image of a generalization called the $j$-fundamental form, see \cite{LM}.

For the scenarios that are considered in this paper, the description of $j$-excitation spaces simplifies, because the weight spaces are one-dimensional. Hence, we can apply the following lemma.
\begin{lemma}
Assume that $\mH$ is a weight-multiplicity-free representation. Then, the $j$-excitation spaces are determined completely by their weights and 
\begin{equation}\label{Nj_weights}
\fl \Lambda_j:= Supp(N_j)= \{\eta\in Supp(\mH) : \eta=\lambda+\alpha_1+\dots+\alpha_j, \alpha_1,\dots,\alpha_j\in\Delta_-^\lambda \} \;.
\end{equation}
\end{lemma} 
\begin{proof}
For weight-multiplicity-free representations weight vectors are in a one-to one correspondence with weights, hence we can consider only weights, and sequences of roots added to the highest weight, instead of sequences of root operators acting on the highest weight vector. The proof will be inductive. Firstly, note that all weights of $N_{j+1}$ are necessarily of the form $\lambda+\alpha_1+\dots+\alpha_j, \alpha_1,\dots,\alpha_{j+1}\in\Delta_-$. This is because applying to $\klambda$ sequences of negative root operators of length smaller than $j+1$ results with vectors from $O^{j}$, which we mod out by definition. Denote 
\[M_{j+1}:=\SpC{\ket{\lambda+\alpha_1+\dots+\alpha_j},\alpha_{j+1}\in\Delta_-}.\]
By the above remark, $j$-excitation spaces are equivalently given by 
\begin{equation}\label{Nj_quotient}
N_{j+1}=M_{j+1}/O^{j}.
\end{equation}
%By the definition of $N_j$, we have $O^j\cong N_j\oplus O^{j-1}$. Hence, $N_{j+1}$ can be realised as quotient $N_{j+1}=M_{j+1}/(N_j\oplus O^{j-1})$. Because $O^{j-1}\cap M_{j+1}=\emptyset$, we actually have 
%\[N_{j+1}=M_{j+1}/N_j.\]
Let us begin with $N_1=O^1/O^0\cong T_{\klambda} G\klambda/\klambda$. By lemma \ref{tangent-space}, we have $N_1=\SpC{E_\alpha\ket{\lambda},\ \alpha\in\Delta_-^\lambda}$, which for weight-multiplicity-free representations means that $N_1=\SpC{\ket{\lambda+\alpha},\ \alpha\in\Delta_-^\lambda}$. Now suppose that $N_{j}$ is spanned by weights of the form (\ref{Nj_weights}). Clearly, the space spanned on the weight vectors with weights from $\Lambda_{j+1}$ is a subspace of $N_{j+1}$. For the sake of contradiction, suppose that in the basis of $N_{j+1}$ there is a weight vector, which is of the form $\ket{\lambda+\alpha_1+\dots+\alpha_{j+1}}$ with $\alpha_i\in\Delta_0^\lambda\cap\Delta_-$ for some $i\leq j+1$. Without any loss of generality we can assume that $i=j+1$. By definition $\ket{\lambda+\alpha_1+\dots+\alpha_{j}}$ belongs to $O^j$ and does not belong to $M_{j+1}$. By the $K_\lambda$-invariance of $O^j$, we have $E_{\alpha_{j+1}}\ket{\lambda+\alpha_1+\dots+\alpha_{j}}\in O^j$, hence by (\ref{Nj_quotient}), 
vector $\ket{\lambda+\alpha_1+\dots+\alpha_{j+1}}$ does not belong to the basis of $N_{j+1}$.
\end{proof}

\begin{remark}
 If $\eta\in\Lambda_j$ and $\xi\in\Lambda_k$ with $|j-k|>1$, then $\eta$ and $\xi$ are root-distinct.
\end{remark}

To determine the local momentum cone at $\lambda$, we are brought to consider the $K_\lambda$-action on the normal space $N_\lambda = \oplus_{j\geq 2} N_j$, and the respective momentum image. This representation is not irreducible except in special cases. It is necessarily reducible, whenever $j_{max}\geq 3$. In section \ref{sec:computational_results} we show how the $j$-excitation spaces decompose further into irreducible components for specific scenarios. As basic tool for the calculation of momentum images of reducible representations in terms of the momentum images of their summands we have the following.

\begin{theorem}[Momentum image for a direct sum of representations \cite{Wildberger}]\label{direct_sum}
Let $\mH_1$ and $\mH_2$ be two unitary representations of $K$ and let $\mH=\mH_1\oplus\mH_2$. Then,
\[\fl\mu(\mH)=\{X\in\mkk:\ X=t X_1+(1-t)X_2,\ 0\leq t\leq1,\ X_1\in\mu(\mH_1),\ X_2\in\mu(\mH_2)\}.\]
In other words, the momentum image of a direct sum of representations is the set of all line segments that connect the momentum images of the components.
\end{theorem}
\begin{proof}
Any normalised vector $\ket{\Psi}\in\mH$ is of the form $\ket{\Psi}=z_1\ket{\Psi_1}+z_2\ket{\Psi_2},\ |z_1|^2+|z_2|^2=1$, for some $\ket{\Psi_i}\in\mH_i$, $\bk{\Psi_i}{\Psi_i}=1$. For any $X\in\mkk$, we have
\begin{eqnarray*}
\i (\mu(\ket{\Psi}),X)=(\overline{z_1}\bra{\Psi_1}+\overline{z_2}\bra{\Psi_2})X(z_1\ket{\Psi_1}+z_2\ket{\Psi_2})=|z_1|^2\bra{\Psi_1}X\ket{\Psi_1}+ \\ +|z_2|^2\bra{\Psi_2}X\ket{\Psi_2},
\end{eqnarray*}
where we used the fact that $\bra{\Psi_1}X\ket{\Psi_2}=\bra{\Psi_2}X\ket{\Psi_1}=0$. In other words,
\[(\mu(\ket{\Psi}),X)=|z_1|^2(\mu(\ket{\Psi_1}),X)+|z_2|^2(\mu(\ket{\Psi_2}),X){\rm\ for\ all\ }X\in\mkk.\]
\end{proof}
\noindent For brevity, we will call such a momentum image the ${\it join}$ of momentum images of the components.
\begin{definition}[Join of momentum images]\label{def:join}
The momentum image of a direct sum of two representations of group $K$ will be called the join of the momentum images of the components.
\begin{eqnarray*}
{\rm Join}(\mu(\mH_1),\mu(\mH_2)):=\{X\in\mkk:\ X=t X_1+(1-t)X_2, \\
 \ 0\leq t\leq1,\ X_1\in\mu(\mH_1),\ X_2\in\mu(\mH_2)\}.
\end{eqnarray*}
\end{definition}
\noindent Despite its simplicity, the above theorem is in general difficult to be applied directly to find the momentum polytope for a direct sum of representations. Note that we have to consider not only the lines joining $\mu(\mH_1)\cap\mkt$ and $\mu(\mH_2)\cap\mkt$ (the so-called momentum rosettes), but also the lines joining points that do not lie in $\mkt$, but intersect $\mkt$ at some point. Such lines are the main source of difficulties. However, in our computations we will consider a cone at $\lambda$, which stems only from one component of the normal space, namely from $N_2$. Remarkably, as we show in Section \ref{sec:computational_results}, in some low-dimensional cases such a cone intersected with the positive Weyl chamber is equal to the whole momentum polytope.

As noted in \cite{Sjamaar}, the local cone at $\lambda$ is contained in the cone spanned by all rays from $\lambda$ to weights from $Supp(N_\lambda)$. Such a cone gives is a crude approximation for the spectral polytope, which boils down to considering the convex hull of local spectra that correspond to computational basis states (weight vectors) that are more than one excitation away from $\ket{\lambda}$.

It is tempting to expect that for some classes of representations the momentum polytope is given just by the intersection of the local cone at $\lambda$ with $\mkt_+$. In such a case, any face of $\poly$ must necessarily contain either the vertex of $\mkt_+$ or the highest weight (or both). However, it turns out that for some physical scenarios the spectral polytope is more complex. For example, the description of the momentum polytope for $\mH_D=\CC^4\otimes\CC^4\otimes\CC^4$ given in \cite{VW} shows that there are some faces that contain neither $\lambda$ nor the vertex of $\mkt_+$. One such face is given by inequality
\begin{eqnarray*} \fl
-5\eta_1^{(1)}-\eta_2^{(1)}+3\eta_3^{(1)}+3\eta_4^{(1)}-5\eta_1^{(2)}+3\eta_2^{(2)}+3\eta_3^{(2)}-\eta_4^{(2)}+5\eta_1^{(3)}+\eta_2^{(3)}+ \\ -3\eta_3^{(3)}-3\eta_4^{(3)}\geq 15.
\end{eqnarray*}
Similarly, for $\Lambda^3(\CC^8)$ an exemplary face that contains neither $\lambda$ nor the vertex of $\mkt_+$ is \cite{Klyachko05}
\[5\eta_1-3\eta_2-3\eta_3+\eta_4+\eta_5+5\eta_6-7\eta_7+\eta_8\geq 3.\]

\section{Spectral polytope from doubly excited states}\label{sec:N2_poly}

Here we define a polytope stemming from doubly excited states, which we later compute for different scenarios in Section \ref{sec:computational_results}. As the central result of this section we show that the spectral polytope stemming from doubly excited states is contained in the whole spectral polytope.

\begin{definition}[Spectral polytope from doubly excited states]\label{N2_poly}
Let $\lambda$ be the highest weight of an irreducible representation of $K$ on $\mH$. Let $N_2\subset\mH$ be the $2$-excitation space from Definition \ref{def:N_j}. Denote the momentum polytope for the representation of the stabiliser of the highest weight, $K_\lambda$, on $N_2$ by 
\[\mP_{K_\lambda}(N_2):=\mu_{K_\lambda}(N_2)\cap(\mkt_\lambda)_+,\]
where $(\mkt_\lambda)_+$ is the positive Weyl chamber of $K_\lambda$. The spectral polytope from doubly excited states, $\mP_2(\mH)$, is the cone at $\lambda$ spanned on $\mP_{K_\lambda}(N_2)$, intersected with $\mkt_+$, i.e.
\[\mP_2(\mH):={\rm Cone}_\lambda(\mP_{K_\lambda}(N_2))\cap\mkt_+.\]
In other words, this is the intersection with $\mkt_+$ of the set of all rays that begin at $\lambda$ and intersect $\mP_{K_\lambda}(N_2)$. 
\end{definition}

It is a nontrivial problem to determine which points of $\mP_2(\mH)$ lie in the image of the momentum map. By definition, points in $\mu_{K_\lambda}(N_2)$ are contained in $\mu_K(\PPH)$. By the convexity of momentum map we have that ${\rm Conv}\left(\mu_{K_\lambda}(N_2)\cap\mkt_+,\lambda\right)$ is contained in $\mu_K(\mH)\cap\mkt_+$. However, the polytope $\mP_2(\mH)$ is larger than this set. In fact, $\mu_{K_\lambda}(N_2)\cap\mkt_+$ can be an empty set, whereas $\mP_2(\mH)$ is nonempty whenever $N_2$ is non-trivial. As an example, consider the representation of $SU(2)\times SU(2)$ on $\mH=S^2\left(\CC^2\right)\otimes S^3\left(\CC^2\right)$, which describes two separated systems of bosons, each system consisting of two modes. The spectral polytope and its relevant parts are shown on Fig.\ref{n2-nontrivial}. 
\begin{figure}[h]
\centering
\includegraphics[width=0.5\textwidth]{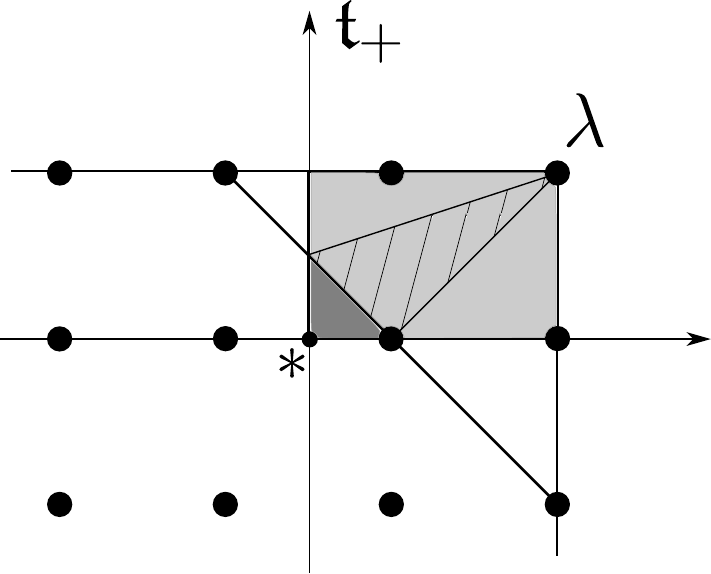}
\caption{Spectral polytope for the system $\mH=S^2\left(\CC^2\right)\otimes S^3\left(\CC^2\right)$. In the mode population notation, the highest weight vector is $\ket{2,0}\otimes\ket{3,0}$. The highest weight lies in the interior of $\mkt_+$, hence $K_\lambda=T$. $N_2=\SpC{\ket{2,0}\otimes\ket{1,2},\ket{1,1}\otimes\ket{2,1},\ket{0,2}\otimes\ket{3,0}}\cong \CC\oplus\CC\oplus\CC$. The image $\mu_{K_\lambda}(N_2)$ is the convex hull of weights from $Supp(N_2)$ - the diagonal solid line. Different parts of momentum polytope that are discussed in the proof of Theorem \ref{N2cone_contained} and in its preceding paragraph are marked with grey.}
\label{n2-nontrivial}
\end{figure}
The extreme weights are root-distinct, hence $\mu_K(\mH)$ is equal to their convex hull. The spectral polytope coincides with $\mP_2(\mH)$. Set ${\rm Conv}\left(\mu_{K_\lambda}(N_2)\cap\mkt_+,\lambda\right)$ is marked as the dashed light grey area. This area is much smaller than the spectral polytope. There is a larger part of the polytope, which is given by just by the positive part of the convex hull of $\mP_{K_\lambda}(N_2))$ and the highest weight (the entire grey area on Fig.\ref{n2-nontrivial}). As we show in the proof of Theorem \ref{N2cone_contained}, such a part is obtained as the local spectra of states from $\mH_\lambda:=\SpC{\ket{\lambda}}\oplus N_2$. However, in order to obtain $\mP_2(\mH)$ one has to add the dark-grey triangle in the corner adjacent to point $*$, which does not come from the momentum image of $\mH_\lambda$. Theorem \ref{N2cone_contained} asserts that there is no such region for the scenarios with fermions and distinguishable particles. However, such a region exists in the 
scenarios involving bosons. Hence, to prove a fact analogous to Theorem \ref{N2cone_contained} 
for bosonic scenarios, we compute the polytopes directly in section \ref{bosons} and compare them with $\mP_2(\mH)$.

\begin{theorem}\label{N2cone_contained}
For representations, which are weight-multiplicity-free and minuscule, the spectral polytope from doubly excited states is contained in the whole spectral polytope, i.e. \[\mP_2(\mH)\subset\poly.\]
Moreover, \[\mP_2(\mH)=\conv{\lambda,\mP_{K_\lambda}(N_2)}\cap\mkt_+.\]
\end{theorem}
\begin{proof}
Recall that polytope $\mP_2(\mH)$ is constructed by considering rays from $\lambda$ to $\mP_{K_\lambda}(N_2)\subset\conv{\Lambda_2}\cap(\mkt_\lambda)_+$. Every such ray is of the form
\begin{equation}\label{ray}
(1-t)\lambda+t\sum_{\eta\in\Lambda_2}a_\eta\eta,\ \sum_{\eta\in\Lambda_2}a_\eta=1,\ a_\eta\geq0,\ t\in\RR_+.
\end{equation}
We will next show that every such ray intersects the positive Weyl chamber at $t\leq 1$. Denote the set of simple roots by $\Pi$. Element of the Cartan algebra $\xi$ is in the positive Weyl chamber if and only if $(\xi,\gamma)\geq 0$ for all $\gamma\in\Pi$. Therefore, ray given by formula (\ref{ray}) intersects the boundary of the positive Weyl chamber if and only if for some $\gamma\in\Pi$ we have
\[(1-t)(\lambda,\gamma)+t\sum_{\eta\in\Lambda_2}a_\eta(\eta,\gamma)=0.\]
Equivalently, 
\[\frac{(\lambda,\gamma)}{(\gamma,\gamma)}-t\sum_{\eta\in\Lambda_2}a_\eta\frac{(\lambda-\eta,\gamma)}{(\gamma,\gamma)}=0.\]
Recall that minuscule representations are such that all the weights are extreme. This means that $\frac{(\lambda,\gamma)}{(\gamma,\gamma)}\in\{0,1\}$ for all $\gamma\in\Pi$. Therefore, we can restrict the set of the considered simple roots to $\Pi-\Pi_0^\lambda$, where $\Pi_0^\lambda=\Pi\cap\Delta_0^\lambda$. By definition of $\Lambda_2$, we have $\lambda-\eta=\alpha+\beta$ for some $\alpha,\beta\in\Delta_+^\lambda$. Recall that for any two roots, we have 
\[\frac{(\alpha,\gamma)}{(\gamma,\gamma)}\in\ZZ.\]
Because $\alpha,\beta$ and $\gamma$ are positive roots, we have
\[\frac{(\alpha+\beta,\gamma)}{(\gamma,\gamma)}\in\ZZ_{\geq 0}.\]
In particular, if there exist $\gamma\in\Pi$ and $\alpha\in\Delta_+^\lambda$ such that $(\alpha,\gamma)\neq0$ (which is a necessary condition for the ray to have a nonzero intersection with the boundary of the positive Weyl chamber), we have $\frac{(\alpha+\beta,\gamma)}{(\gamma,\gamma)}\geq1$ for all $\beta\in\Delta_+^\lambda$. This in turn implies that
\[\sum_{\eta\in\Lambda_2}a_\eta\frac{(\lambda-\eta,\gamma)}{(\gamma,\gamma)}\geq1,\]
hence the ray intersects the boundary of the positive Weyl chamber at 
\[t=\left(1-\sum_{\eta\in\Lambda_2}a_\eta\frac{(\eta,\gamma)}{(\gamma,\gamma)}\right)^{-1}\leq1.\]
In the above formula, we used the fact that $(\lambda,\gamma)/(\gamma,\gamma)=1$. Therefore, the intersection of the $N_2$-cone with the positive Weyl chamber takes place within the convex hull of $\mP_{K_\lambda}(N_2)$ and the highest weight, i.e.
\[\mP_2(\mH)=\conv{\lambda,\mP_{K_\lambda}(N_2)}\cap\mkt_+.\]
Consider the space $\mH_\lambda:=\SpC{\ket{\lambda}}\oplus N_2$. We will show that $\mu_K(\mH_\lambda)= \mP_2(\mH)$. Write vectors from $\mH_\lambda$ as $\ket{\Psi}=a\ket{\lambda}+\ket{\phi}$, where $\ket{\phi}\in N_2$ and $a\in\CC$. Then, we have
\[\mu_K(\kpsi)=\frac{1}{|a|^2+\bk{\phi}{\phi}}\left(|a|^2\lambda+\mu_{K_\lambda}(\ket{\phi})\right).\]
To see this, consider the non-diagonal components of the momentum image of $\kpsi$, i.e. $(\mu(\kpsi),E_\alpha)$.
\begin{equation}\label{maps_relation}
\fl\bra{\Psi}E_\alpha\kpsi=|a|^2\bra{\lambda}E_\alpha\klambda+\overline a\blambda E_\alpha\ket{\phi}+a\bra{\phi} E_\alpha\klambda+\bra{\phi}E_\alpha\ket{\phi}=\bra{\phi}E_\alpha\ket{\phi},
\end{equation}
where we used the fact that $\bra{\lambda}E_\alpha\klambda\propto\bk{\lambda}{\lambda+\alpha}=0$ and $\blambda E_\alpha\ket{\phi}=0$, because $E_\alpha\ket{\phi}\in N_1\oplus N_2\oplus N_3$. Moreover, if $\alpha\in\Delta-\Delta_0$, then $E_\alpha\ket{\phi}\in N_1\oplus N_2$, hence in such a case we have $\bra{\phi}E_\alpha\ket{\phi}=0$. This means that the non-diagonal components of $\mu_K$ and $\mu_{K_\lambda}$ are the same. For the diagonal components, we obtain formula (\ref{maps_relation}) by writing $\ket{\phi}$ as a combination of weight vectors from $N_2$. Clearly, by choosing different values of parameter $a$ and all possible vectors $\ket{\phi}$ mapped to $\mkt_+$, we obtain all points from the convex hull of $\lambda$ and $\mu_{K_\lambda}(N_2)\cap\mkt_+$. 
\end{proof} 

As we show in theorem \ref{thm:spherical}, the polytope $\mP_2(\mH)$ is equal to the momentum polytope for a certain class of representations, called {\it spherical} representations. A representation is spherical, when every fibre of the momentum map $\mu^{-1}(\xi)$, $\xi\in\mkt$, is a single $K$-orbit, or, equivalently, when the Borel subgroup $B$ of $G$ has an open orbit in the projective space $\PP(\mH)$, \cite{Brion1987,Kac-nilp-orb}. A classification of spherical representations is given in \cite{Knop-MultFreeSpa}. In particular, all two-particle scenarios are spherical \cite{HKS13}.

\begin{theorem}\label{thm:spherical}
Assume $\mH$ is a spherical representation. Then the second osculating space of the coherent orbit fills the entire representation space, i.e. $O^2_{|\lambda\rangle}G|\lambda\rangle = \mH$. In other words, the second fundamental form is surjective onto the normal space. Consequently, the polytope from doubly excited states equal the momentum polytope, i.e. $\poly=\mP_2(\mH)$.
\end{theorem}
\begin{proof}
By a direct inspection of the list of spherical representations, we note that for such representations $N_\lambda=N_2$. Then, ${\rm Cone}_\lambda(\mP_{K_\lambda}(N_2))$ is the local cone at $\lambda$, hence by equation (\ref{hw_cone_containing_poly}) $\poly\subset\mP_2(\mH)$. On the other hand, we have $\mP_2(\mH)\subset\poly$ by the same reasoning as in the proof of theorem \ref{N2cone_contained}. Namely, for subspace $\mH_\lambda=\SpC{\ket{\lambda}}\oplus N_2\subset\mH$, we have $\mu_K(\mH_\lambda)= \mP_2(\mH)$.
\end{proof}

%\begin{remark}
%While sphericity is a property of the group action, the geometric constructions
%such as tangent spaces or osculating spaces to the locus of non-entangled states 
%$\hat M = G\vert\lambda\rangle\subset\mH$ depend only on $\hat M$ as a submanifold of $\mH$.
%\end{remark}

\section{Applications and examples}\label{sec:computational_results}
In this section we compute all spectral polytopes, which are given by the polytopes from doubly entangled states. We also give a simple example of a system, where the polytope is larger than the one stemming from the doubly excited states.

\subsection{Distinguishable particles}
Consider first the natural representation of $U(N)$ on $\CC^N$, i.e. a one particle scenario. The spectral polytope is just a point - the highest weight, i.e.
\[\mP(\CC^N)=(1,0,\dots,0).\]
The stabiliser of the highest weight is $U(1)\times U(N-1)$, where $U(1)$ acts by scalars on $\ket{1}$ and $U(N-1)$ acts on the space $\SpC{\ket{2},\ket{3},\dots,\ket{N}}$. The positive root operators of the stabiliser are $\kb{i}{j},\ 2\leq i<j\leq L$. 

For a system of distinguishable particles, the stabiliser of the highest weight is equal to the product of one-particle stabilisers
\[K_\lambda=\bigtimes_{k=1}^LU(1)\times U(N_k-1).\]
The set of positive root operators from $K_\lambda$ is of the following form
\[\bone\otimes\dots\bone\otimes \kb{i_k}{j_k}\otimes\bone\otimes\dots\otimes\bone,\ 2\leq i_k<j_k\leq N_k, 1\leq k\leq L.\] 
For the single particle scenario, the only nontrivial excitation spaces are $N_0=\SpC{\ket{1}}$ and $N_1=\SpC{\ket{2},\dots,\ket{N}}$. It is easy to see, as the set of negative root operators for roots that do not belong to $\Delta_0$ is equal to $\{\kb{n}{1},\ 2\leq n\leq N\}$.  For finding the excitation spaces of tensor products of representations, we will use the fact that the number of excitations adds under the tensor product.
\begin{lemma}\cite{LM}
Let $\mH=\mH_1\otimes\mH_2$ be the tensor product of two representations. Denote by $N_j^{(m)}$ the $j$-excitation space of $\mH_m$, $m=1,2$. Then, we have
\[N_m^{(1)}\otimes N_n^{(2)}\subset N_{m+n},\ N_j=\bigoplus_{m+n=j}N_m^{(1)}\otimes N_n^{(2)}.\]
\end{lemma}
\noindent Therefore, for a system of $L$ particles we have
\[N_j=\bigoplus_{1\leq k<l\leq L}\left(\bigoplus_{m+n=j}N_m^{(k)}\otimes N_n^{(l)}\right).\]
In particular, for $N_2$ we only need to consider products of $1$-excitation spaces of the components, i.e.
\[\fl N_2=\bigoplus_{1\leq k<l\leq L}\CC^{(N_k-1)}\otimes \CC^{(N_l-1)}=\bigoplus_{1\leq k<l\leq L}\SpC{\ket{m}^{(k)}\otimes\ket{n}^{(l)},\ 2\leq m\leq N_k,\ 2\leq n\leq N_l}.\]

\paragraph*{Two particles} Consider the general two-particle scenario, i.e. $\mH_D=\CC^{M}\otimes \CC^{N},\ N>M$. Recall that the spectral polytope is in this case given by inequalities
\begin{equation}\label{2part_solution}
\fl \eta_1^{(1)}\geq \eta_2^{(1)}\geq\dots\geq\eta_M^{(1)}\geq0,\ \eta^{(2)}_i=\eta^{(1)}_i{\rm\ for\ }1\leq i\leq M,\ \eta^{(2)}_i=0{\rm\ for\ }M+1\leq i\leq N,
\end{equation}
plus the normalisation condition $\sum_{i=1}^M\eta_i^{(1)}=1$.
We will show that these inequalities can be reproduced just by considering the polytope from doubly excited states. The first inequalities are just the inequalities for the positive Weyl chamber of $U(M)$, so the difficulty lies in showing that  $ \eta^{(2)}_i=\eta^{(1)}_i{\rm\ for\ }1\leq i\leq M$ and $\eta^{(2)}_i=0{\rm\ for\ }M+1\leq i\leq N$. Firstly, note that in the two-particle scenario, we have $N_j=0$ for $j\geq 3$. Space $N_2$ is of the form
\[N_2=\CC^{(M-1)}\otimes \CC^{(N-1)}=\SpC{\ket{m}\otimes\ket{n},\ 2\leq m\leq M,\ 2\leq n\leq N}.\]
Therefore, we are coming back to the problem of describing the momentum polytope, but for smaller one-particle spaces. We will next proceed inductively and consider the sequence of $2$-excitation spaces nested in each other, i.e. $N_2\supset N_2'\supset N_2''\supset\dots\supset N_2^{\rm final}=\CC\otimes\CC^{(N-M+1)}$. We do it by picking the highest weight in each $2$-excitation space and considering its stabiliser contained in the stabiliser of the highest weight from the preceding space. For example, the highest weight in $N_2$ is $\ket{\lambda'}=\ket{2}\otimes\ket{2}$, $K_{\lambda'}=U(1)\times U(M-1)\times U(1)\times U(N-1)\subset K_\lambda$ and the positive root operators of $K_{\lambda'}$ are
\[\kb{i_1}{j_1}\otimes\bone,\ \bone\otimes\kb{i_2}{j_2},\ 3\leq i_1<j_1\leq M,\ 3\leq i_2<j_2\leq N.\]
Proceeding in this way, we end up with representation $N_2^{\rm final}$ of $U(1)\times U(1)\times U(N-M+1)$, whose momentum image is just the properly shifted momentum image for the one-particle case of size $N-M+1$.
\begin{eqnarray*}
\ket{\lambda'}=\ket{2}\otimes\ket{2},\ N_2'=\SpC{\ket{m}\otimes\ket{n},\ 3\leq m\leq N_k,\ 3\leq n\leq N_l}, \\
\vdots \\
\ket{\lambda^{\rm final}}=\ket{M-1}\otimes\ket{M-1},\ N_2^{\rm final}=\SpC{\ket{M}\otimes\ket{n},\ M\leq n\leq N}.
\end{eqnarray*}
The representation $N_2^{\rm final}$ yields a polytope, which is just a point - the momentum image if the highest weight, i.e.
\[\mP_2^{\rm final}=\mu(\ket{M}\otimes\ket{M})=((0,\dots,0,1),(0,\dots,0,0,1,0,\dots,0))\]
where $1$ is on the $M$-th coordinate in both vectors. According to the definition (\ref{N2_poly}), the spectral polytope of the penultimate $2$-excitation space is the cone at $\lambda^{\rm final}=((0,\dots,1,0),(0,\dots,0,1,0,0,\dots,0))$
spanned by $\mP_2^{\rm final}$. This is the line
\[((0,\dots,1-t_1,t_1),(0,\dots,0,1-t_1,t_1,0,\dots,0)),\ t\in\RR_+.\]
The intersection with the proper positive Weyl chamber imposes the conditions that $1-t_1\geq t_1$. Hence, the penultimate spectral polytope is described by $t_1$ ranging from $0$ to $1/2$. In this way, we recover condition $\eta^{(2)}_{M-1}=\eta^{(1)}_{M-1},\ \eta^{(2)}_{M}=\eta^{(1)}_{M},\ \eta^{(2)}_i=0{\rm\ for\ }M+1\leq i\leq N$ and $\eta^{(1)}_{M-1}\geq\eta^{(1)}_{M}$. As a next step, we consider the cone at weight $\mu(\ket{M-2}\otimes\ket{M-2})$ spanned by the penultimate spectral polytope. This is the set described by
\begin{eqnarray*}
\fl \eta^{(1)}_{M-2}=1-t_2,\ \eta^{(1)}_{M-1}=t_2(1-t_1),\ \eta^{(1)}_{M}=t_2t_1,\ \eta^{(1)}_{i}=0{\rm\ for\ }1\leq i\leq M-3, \\
\eta^{(2)}=(\eta^{(1)},0,\dots,0),\ 0\leq t_1\leq\frac{1}{2},\ t_2\in\RR_+.
\end{eqnarray*}
Again, the intersection with the positive Weyl chamber gives $0\leq t_2\leq\frac{1}{(2-t_1)}$. At the end of this procedure, we get a polytope parametrised by $t_1,t_2,\dots,t_{M-1}$, $0\leq t_i\leq\frac{1}{(2-t_{i-1})}$, $0\leq t_1\leq\frac{1}{2}$, whose points are
\begin{eqnarray*}
\fl \eta^{(1)}_{1}=1-t_{M-1},\ \eta^{(1)}_{2}=t_{M-1}(1-t_{M-2}),\ \eta^{(1)}_{3}=t_{M-1}t_{M-2}(1-t_{M-3}),\dots, \\ \fl \eta^{(1)}_{M}=\prod_{i=1}^{M-1}t_i,\ \eta^{(2)}=(\eta^{(1)},0,\dots,0).
\end{eqnarray*}
Note that the normalisation $\sum_{i=1}^M\eta_i^{(1)}=1$ is satisfied automatically. Clearly, inequalities $(\ref{2part_solution})$ are satisfied by the above points. They can be saturated by choosing $t_i=\frac{1}{(2-t_{i-1})}$ for some $i$. Therefore, the polytope given by inequalities $(\ref{2part_solution})$ is identical to the polytope defined via the above parametrisation.

\paragraph*{$L$ qubits} The spectral polytope for a system of $L$ qubits is given by the inequalities
\begin{equation}\label{qubits_solution}
\eta^{(k)}_{2}\leq\sum_{l\neq k}\eta^{(l)}_{2}{\rm\ and\ }\eta^{(k)}_{1}\geq\eta^{(k)}_{2}{\rm\ for\ }1\leq k\leq L
\end{equation}
with normalisation $\eta^{(k)}_{1}+\eta^{(k)}_{2}=1$, cf. \cite{HSS03}. Let us next show that these inequalities can be obtained from the spectral polytope stemming from doubly excited states. System of $L$-qubits is the only scenario, where the highest weight is regular, i.e. lies in the interior of the positive Weyl chamber. This means that the stabiliser of the highest weight is equal to the maximal torus, $K_\lambda=T\subset K$. The momentum image stemming from an action of a torus is easy to compute, as it is just the convex hull of weights (see lemma \ref{lemma:torus_poly}). We will next show that inequalities $\eta^{(k)}_{2}\leq\sum_{l\neq k}\eta^{(l)}_{2}$ describe the cone at the highest weight. Denote the weights in $\Lambda_2$ by $\lambda_{i,j}$. They are of the form (we omit the tensor product symbols)
\begin{eqnarray*}
\fl \lambda_{i,j}:=\mu_T(\ket{1}\dots\ket{1}\ket{2}^{(i)}\ket{1}\dots\ket{1}\ket{2}^{(j)}\ket{1}\dots\ket{1})= \\ =((1,0),\dots,(1,0),(0,1),(1,0),\dots,(1,0),(0,1),(1,0),\dots,(1,0)).
\end{eqnarray*}
According to lemma \ref{lemma:torus_poly}, the polytope $\mP_{K_\lambda}(N_2)$ is the convex hull of $\Lambda_2$. By a straightforward calculation, one can check that for a convex combination of weights $\sum_{i<j}a_{i,j}\lambda_{i,j}$, $\sum_{i<j}a_{i,j}=1$, $a_{i,j}\geq 0$ we have
\[\eta_2^{(1)}=\sum_{j=2}^{L}a_{1,j},\ \eta_2^{(k)}=\sum_{j=1}^{k-1}a_{j,k}+\sum_{j=k+1}^{L}a_{k,j},\ 2\leq k\leq L-1,\ \eta_2^{(L)}=\sum_{j=1}^{L-1}a_{j,L}.\]
Any point from the $N_2$-cone is of the form
\[(1-t)\lambda+t\sum_{i<j}a_{i,j}\lambda_{i,j}, t\in\RR_+.\]
Because $\lambda=((1,0),\dots,(1,0))$, the $\eta_2^{(l)}$-coordinates of the cone are just the $\eta_2^{(l)}$-coordinates of the convex hull of $\Lambda_2$ multiplied by $t$. Consider the expression $\left(\sum_{l\neq k}t\eta^{(l)}_{2}\right)-t\eta_2^{(k)}$. Since all $a_{i,j}$s add up to one, we have $\sum_{l=1}^L\eta_2^{(l)}=2$. Hence, 
\[\left(\sum_{l\neq k}t\eta^{(l)}_{2}\right)-t\eta_2^{(k)}=2t(1-\eta_2^{(k)}).\]
The above expression is always non-negative, as $\eta_2^{(k)}$ is a sum of some subset of $a_{i,j}$s. Moreover, the above expression is equal to zero if and only if the only nonzero coefficients are $a_{j,k},\ j<k$ and $a_{k,j},\ j>k$ (for the considered $k$). Such convex combinations of weights describe the external faces of the $N_2$-cone. Hence, we have shown that points from intersection of the $N_2$-cone and $\mk t_+$ can saturate the inequalities for the spectral polytope, which means that the polytopes are identical. See figure \ref{3qubits} to see how the intersection of cones gives the spectral polytope for $3$ qubits.
\begin{figure}[h]
\centering
\includegraphics[width=0.7\textwidth]{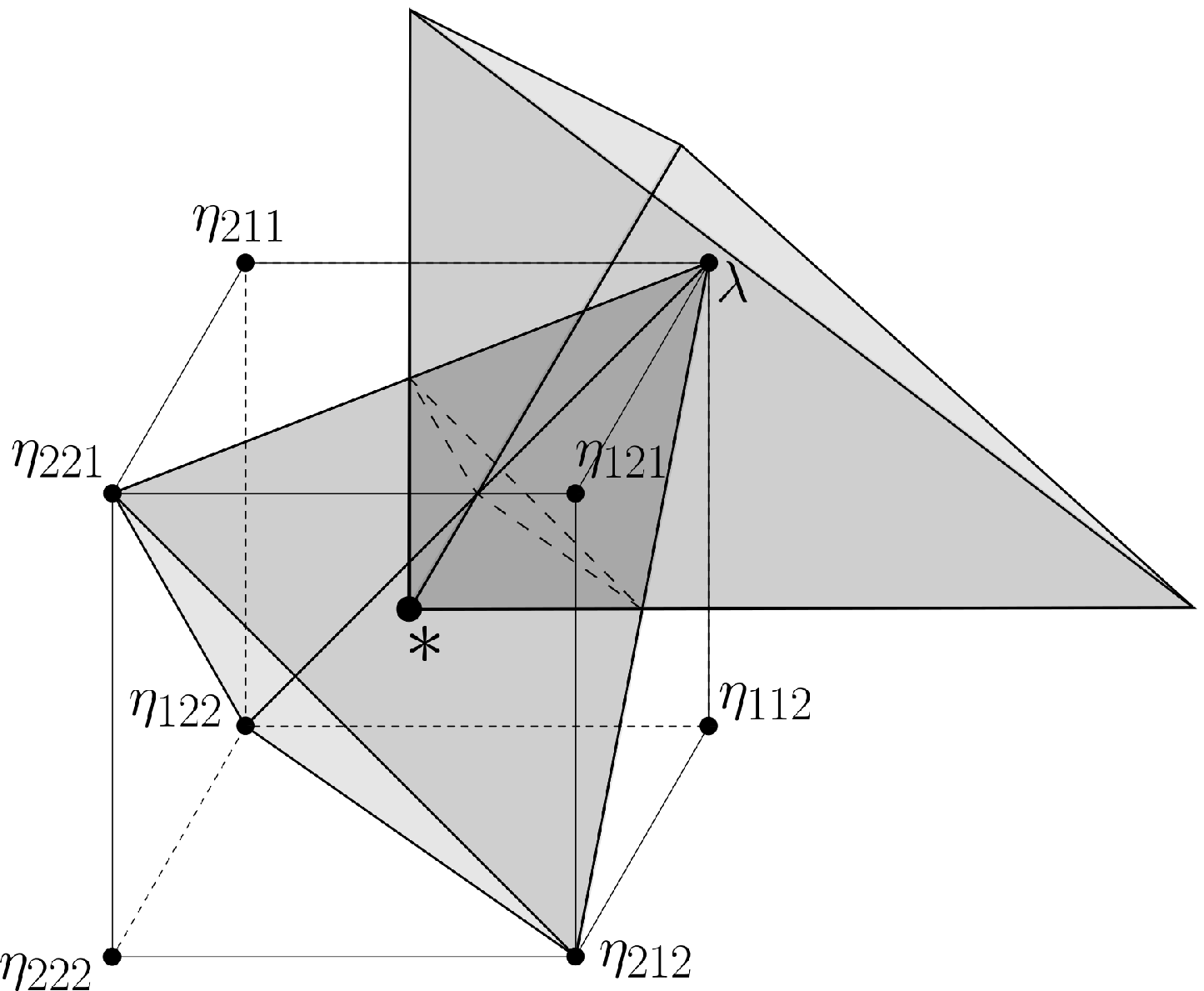}
\caption{The spectral polytope for three qubits as an intersection of two cones. The cone at vertex $*$ is the positive Weyl chamber, i.e. the sector, where the eigenvalues of the one-qubit reduced density matrices are ordered decreasingly. The cone at $\lambda$ is the local cone at the highest weight. The dashed line denotes the intersection of cones' boundaries.}
\label{3qubits}
\end{figure}

\paragraph*{System of size $2\times2\times 3$} 
As an example, where the polytope from doubly excited states is smaller than the whole spectral polytope, we consider the $2\times2\times 3$ system. This example also illustrates how to compute the the momentum image of a direct sum of representations using Theorem \ref{direct_sum}. The two-excitation space is 
\[N_2=\left(\CC\otimes\CC\right)\oplus\left(\CC\otimes\CC^2\right)\oplus\left(\CC\otimes\CC^2\right),\]
where the components are the irreducible components of the representation of $K_\lambda=(U(1)\times U(1))\times (U(1)\times U(1))\times (U(1)\times U(2))$ on $N_2$. They are spanned by weight vectors corresponding to weights
\begin{eqnarray*}
\eta_{221}:=\mu(\ket{2}\ket{2}\ket{1})=((0,1),(0,1),(1,0,0)), \\
\eta_{122}=((1,0),(0,1),(0,1,0)),\ \eta_{123}=((1,0),(0,1),(0,0,1)), \\
\eta_{212}=((0,1),(1,0),(0,1,0)),\ \eta_{213}=((0,1),(1,0),(0,0,1))
\end{eqnarray*}
respectively. The join of $\mu\left(\CC\otimes\CC\right)$ with $\mu\left(\left(\CC\otimes\CC^2\right)\oplus\left(\CC\otimes\CC^2\right)\right)$ intersected with $\mkt_\lambda$ is easy to compute, because $\mu\left(\CC\otimes\CC\right)$ is a single point, namely $\eta_{221}$. Hence, the result is
\begin{eqnarray*}
 \fl {\rm Join}\left(\mu\left(\CC\otimes\CC\right),\mu\left(\left(\CC\otimes\CC^2\right)\oplus\left(\CC\otimes\CC^2\right)\right)\right)\cap(\mkt_\lambda)_+= \\ =\conv{\eta_{221},\mu\left(\left(\CC\otimes\CC^2\right)\oplus\left(\CC\otimes\CC^2\right)\right)\cap(\mkt_\lambda)_+}.
 \end{eqnarray*}
 For the image of $\left(\CC\otimes\CC^2\right)\oplus\left(\CC\otimes\CC^2\right)$, we have the following result.
 \begin{lemma}
 The $K_\lambda$-momentum polytope of $\left(\CC\otimes\CC^2\right)\oplus\left(\CC\otimes\CC^2\right)$ is given by
 \[ \fl {\rm Join}\left(\mu\left(\CC\otimes\CC^2\right),\mu\left(\CC\otimes\CC^2\right)\right)\cap(\mkt_\lambda)_+=\conv{\eta_{122},\eta_{212},\frac{1}{2}(\eta_{212}+\eta_{123})}.\]
 \end{lemma}
 \begin{proof}
 The intersection of the momentum image of a representation with the torus is contained in the convex hull of weights 
 \[\mu_{K_\lambda}\left(\left(\CC\otimes\CC^2\right) \oplus\left(\CC\otimes\CC^2\right)\right) \cap\mk{t} \subset\conv{\eta_{122},\eta_{123},\eta_{212},\eta_{213}},\]
 The above convex hull is two-dimensional (it is depicted on Fig.\ref{223_cartan}). 
 \begin{figure}[h]
\centering
\includegraphics[width=0.35\textwidth]{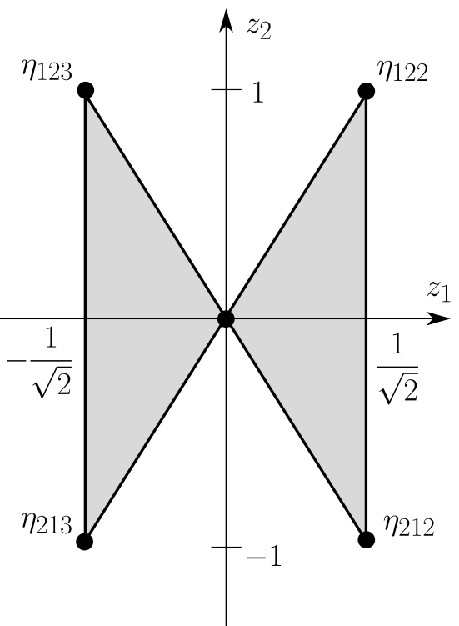}
\caption{The join of momentum images for $\left(\CC\otimes\CC^2\right)\oplus \left(\CC\otimes\CC^2\right)$.}
\label{223_cartan}
\end{figure}
 The positive Weyl chamber is the region $z_1\geq 0$. The whole considered algebra is four-dimensional, where the two additional dimensions come from the non-diagonal components of $\mku(2)$. The momentum images of the components are the properly shifted momentum images for the $U(2)$ action on $\CC^2$, which are known to be two-dimensional spheres. The spheres are centered at $z_2=\pm 1,\ z_1=0$ and have radii
 \[r=\frac{1}{2}|\eta_{122}-\eta_{123}|=\frac{1}{2}|\eta_{212}-\eta_{213}|=\frac{\sqrt{2}}{2}.\]
The points on the spheres are of the form
\[(u,\pm1,v,w),\  u^2+v^2+w^2=r^2,\]
where the first two coordinates are $z_1,\ z_2$. In order to compute the join of the spheres, we consider the line that starts at point $p_1$ on the sphere centered at $z_2=-1$ and contains a point $p_2$ from the Cartan algebra, i.e.
 \[\fl l:\ (1-t)p_2+t p_2,\ t\in\RR_+,\ p_1=(u,-1,v,w),\ u^2+v^2+w^2=r^2,\ p_2=(z_1,z_2,0,0).\]
 Line $l$ intersects the hyperplane $z_2=1$ at $t$ such that $(-1)(1-t)+z_2t=1$, i.e. at $t=\frac{2}{1+z_2}$. For such a value of parameter $t$, the intersection point reads
 \[p=\frac{1}{1+z_2}\left((z_2-1)u+2z_1,1,v,w\right).\]
 The condition for $p$ to belong to the second sphere reads 
 \[\frac{1}{(1+z_2)^2}\left(((z_2-1)u+2z_1)^2+(v^2+w^2)(z_2-1)^2\right)=r^2,\]
 which yields the following equation for $z_2$ as a function of $z_1$ and $u$.
 \[z_2=\frac{z_1(z_1-u)}{r^2-z_1u},\ u\in]-r,r[,\ z_2=\pm\frac{1}{r}z_1{\rm\ for}\ u=\pm r.\]
 It is straightforward to check that the above family of curves fills the shaded region on Fig.\ref{223_cartan}.
 \end{proof}
 \noindent In order to find the polytope generated from doubly excited states, one has to compute the intersection
 \[\mc{P}_2=\conv{\eta_{122},\eta_{212},\frac{1}{2}(\eta_{212}+\eta_{123}),\eta_{221},\lambda}\cap\mkt_+.\]
 The result reads 
\begin{eqnarray*}
\fl  \mc{P}_2={\rm Conv}{\Bigg(}\lambda,\left(\left(\frac{1}{2},\frac{1}{2}\right),\left(\frac{1}{2},\frac{1}{2}\right),\left(1,0,0\right)\right),\left(\left(\frac{1}{2},\frac{1}{2}\right),\left(1,0\right),\left(\frac{1}{2},\frac{1}{2},0\right)\right), \\ \left(\left(1,0\right),\left(\frac{1}{2},\frac{1}{2}\right),\left(\frac{1}{2},\frac{1}{2},0\right)\right),\left(\left(\frac{1}{2},\frac{1}{2}\right),\left(\frac{1}{2},\frac{1}{2}\right),\left(\frac{1}{3},\frac{1}{3},\frac{1}{3}\right)\right), \\ \left(\left(\frac{2}{3},\frac{1}{3}\right),\left(\frac{2}{3},\frac{1}{3}\right),\left(\frac{1}{3},\frac{1}{3},\frac{1}{3}\right)\right){\Bigg)}.
\end{eqnarray*}
Comparing the vertices of polytope $\mc{P}_2$ and the vertices of the whole momentum polytope computed in \cite{VW}, one can see that 
\begin{eqnarray*}\fl\poly=\conv{\mc{P}_2,\left(\left(\frac{1}{2},\frac{1}{2}\right),\left(\frac{3}{4},\frac{1}{4}\right),\left(\frac{1}{2},\frac{1}{4},\frac{1}{4}\right)\right),\left(\left(\frac{3}{4},\frac{1}{4}\right),\left(\frac{1}{2},\frac{1}{2}\right),\left(\frac{1}{2},\frac{1}{4},\frac{1}{4}\right)\right)}.
\end{eqnarray*}

\subsection{Fermions}\label{fermions}
Consider $L$ fermions on $N$ modes. The stabiliser of the highest weight is $K_\lambda=U(L)\times U(N-L)$, where $U(L)$ acts on the first $L$ basis vectors and $U(N-L)$ acts on the last $N-L$ basis vectors of $\CC^N$. The positive root operators of $\mkk^\CC_\lambda$ are $E_{i,j}$, where $1\leq i<j\leq L$ or $L+1\leq i <j\leq N$. The $j$-excitation spaces have been computed for this scenario in \cite{LM} and they are of the form
\begin{equation}\label{Nj_fermions}
N_j=\Lambda^{L-j}(\CC^{L})\otimes\Lambda^{j}(\CC^{N-L}),\ j\in\{0,1,\dots L\}.
\end{equation}

\paragraph*{Two fermions}
The case of two fermions is described by a spherical representation, which means that the polytope from doubly excited states is equal to the whole polytope. For two fermions in $N$ modes, we have
\[N_2=\CC\otimes\Lambda^2(\CC^{N-2})=\SpC{\ket{i}\wedge\ket{j}:\ 3\leq i<j<N}.\]
Similarly, as in the case of two distinguishable particles, we will iterate the procedure to find the momentum polytope of $N_2$, which is the convex hull of $\ket{\lambda'}=\ket{3}\wedge\ket{4}$ and the momentum polytope of
\[N_2'=\CC\otimes\Lambda^2(\CC^{N-4})=\SpC{\ket{i}\wedge\ket{j}:\ 5\leq i<j<N}.\]
In the last iteration, for $N$ even we end up with 
\[\ket{\lambda^{\rm final}}=\ket{N-3}\wedge\ket{N-2},\ N_2^{\rm final}=\CC=\SpC{\ket{N-1}\wedge\ket{N}},\]
end for $N$ odd we have
\begin{eqnarray*}
\ket{\lambda^{\rm final}}=\ket{N-4}\wedge\ket{N-3},\\ N_2^{\rm final}=\CC^3=\SpC{\ket{N-2}\wedge\ket{N-1},\ket{N-2}\wedge\ket{N},\ket{N-1}\wedge\ket{N}}.
\end{eqnarray*}
We compute the polytopes similarly to the case of two distinguishable particles. In the case of even $N$, the polytope from the last step is the line segment
\[(1-t_1)(0,\dots,0,1,1,0,0)+t_1(0,\dots,0,0,0,1,1),\ 0\leq t_1\leq \frac{1}{2}.\]
The penultimate polytope is parametrised by
\[\fl(0,\dots,0,1-t_2,1-t_2,t_2(1-t_1),t_2(1-t_1),t_2t_1,t_2t_1),\ 0\leq t_2\leq \frac{1}{2-t_1},\ 0\leq t_1\leq \frac{1}{2}.\]
Finally, for $N=2K$, we get
\begin{eqnarray*}
\fl \eta_{2k+1}=\eta_{2k+2},\ \eta_{2k+1}=(1-t_{K-k})\prod_{l=K-k+1}^Kt_l\ {\rm for\ }k=0,1,\dots,K-2,\\\fl \eta_{2K-1}=\eta_{2K}=\prod_{l=1}^Kt_k,\ 0\leq t_{i+1}\leq \frac{1}{2-t_i},\ 0\leq t_1\leq \frac{1}{2}.
\end{eqnarray*}
For $N$ odd we put $K=(N-1)/2$ and $\eta_N=0$. The obtained parametrisation is essentially a reformulation of the well-known inequalities for the polytope, which are $\eta_{2k+1}=\eta_{2k+2}$, $k=0,1,\dots,K-1$ and $\eta_{1}\geq\eta_{2}\geq\dots\eta_N$.

\paragraph*{Three fermions}	
The first fermionic scenario, where the spectral polytope is not given only by equalities (except for the inequalities for the positive Weyl chamber), are three fermions on six levels. In this subsection, we will also consider the scenario with seven levels. These scenarios complete the list of all fixed-particle scenarios, where the spectral polytope is given by the polytope from doubly excited states. 
Let us begin with $\mH_F=\Lambda^3(\CC^6)$. The highest weight is $\ket{\lambda}=\ket{1}\wedge\ket{2}\wedge\ket{3}$. According to equation (\ref{Nj_fermions}), the space of doubly excited states in this case is of the form 
\[\fl N_2=\CC^3\otimes\Lambda^2(\CC^3)\cong \CC^3\otimes\CC^3=\SpC{\ket{a}\wedge\ket{i}\wedge\ket{j}:\ a\in\{1,2,3\},\ 4\leq i<j\leq 6}.\]
Let us denote by $\eta_{ijk}$ the weight vector $\ket{i}\wedge\ket{j}\wedge\ket{k}$. The polytope corresponding to the action of $K_\lambda$ on $N_2$ is the spectral polytope for two distinguishable particles on $3$ levels, which is 
\[\mu_{K_\lambda}(N_2)\cap(\mkt_\lambda)_+=\conv{\eta_{145},\eta_{246},\eta_{356}}\cap(\mkt_\lambda)_+.\]
The intersection with $\mkt_+$ of the cone at $\lambda$ spanned by $\mu_{K_\lambda}(N_2)\cap(\mkt_\lambda)_+$ is of the following form
\begin{eqnarray*}
\mP_2=\conv{\eta_{123},\eta_{145},\eta_{246},\eta_{356}}\cap\mkt_+= \\ \fl = \conv{(1,1,1,0,0,0),\left(1,\frac{1}{2},\frac{1}{2},\frac{1}{2},\frac{1}{2},0\right),\left(\frac{3}{4},\frac{3}{4},\frac{1}{2},\frac{1}{2},\frac{1}{4},\frac{1}{4}\right),\left(\frac{1}{2},\frac{1}{2},\frac{1}{2},\frac{1}{2},\frac{1}{2},\frac{1}{2}\right)}.
\end{eqnarray*}
In can be checked by a direct computation that the well-known inequalities for the polytope that can be found in \cite{BD,Klyachko04,Ruskai07} yield the same set of vertices. It is quite surprising that even though the $N_3$ space is nontrivial (it is one-dimensional), the whole polytope can be obtained only by considering the doubly excited states. By a similar strategy, we next show that this is also the case for $\Lambda^3(\CC^7)$. The doubly excited states are in this scenario of the following form.
\[N_2=\CC^3\otimes\Lambda^2(\CC^4)=\SpC{\ket{a}\wedge\ket{i}\wedge\ket{j}:\ a\in\{1,2,3\},\ 4\leq i<j\leq 7}.\]
In order to find the polytope $\mu_{K_\lambda}(N_2)\cap(\mkt_\lambda)_+$, we repeat the procedure for the highest weight of $K_\lambda$ in $N_2$, which is $\ket{\lambda'}=\ket{1}\wedge\ket{4}\wedge\ket{5}$. Then, space $N_2'$ is the sum of tensor products of proper $N_j$-spaces of the components
\begin{eqnarray*}
N_2'=\left(N_1(\CC^3)\otimes N_1(\Lambda^2(\CC^4))\right)\oplus\left(N_0(\CC^3)\otimes N_2(\Lambda^2(\CC^4))\right)= \\ =\left(\CC^2\otimes\CC^2\otimes\CC^2\right)\oplus\left(\CC\otimes\CC\otimes\CC\right).
\end{eqnarray*}
The first component is just the $3$-qubit polytope and the second component is the complex span of $\ket{1}\wedge\ket{6}\wedge\ket{7}$. Hence, in order to compute $\mu_{K_\lambda'}(N_2)\cap(\mkt_\lambda')_+$, which, by Theorem \ref{direct_sum}, is the join of the momentum images of the components, it is enough to take the convex hull of the $3$-qubit polytope with $\eta_{167}$. This is because the momentum image of the second component is just a single point. The $3$-qubit polytope is spanned by the vertices
\[\eta_{246},\frac{1}{2}(\eta_{246}+\eta_{257}),\frac{1}{2}(\eta_{246}+\eta_{347}),\frac{1}{2}(\eta_{246}+\eta_{356}),\frac{1}{2}(\eta_{246}+\eta_{357}).\] 
The intersection of the convex hull of the above vertices and weights $\eta_{167}$, $\eta_{145}$, $\lambda=\eta_{123}$ with the positive Weyl chamber has the following vertices:
\begin{eqnarray*}
\lambda, \left(\frac{3}{7},\frac{3}{7},\dots,\frac{3}{7},\frac{3}{7}\right),\left(\frac{2}{3},\frac{2}{3},\frac{1}{3},\dots,\frac{1}{3}\right),\left(\frac{5}{7},\frac{5}{7},\frac{3}{7},\frac{3}{7},\frac{3}{7},\frac{1}{7},\frac{1}{7}\right), \\ \left(\frac{1}{2},\frac{1}{2},\frac{1}{2},\frac{1}{2},\frac{1}{2},\frac{1}{4},\frac{1}{4}\right),
\left(\frac{3}{5},\frac{3}{5},\frac{3}{5},\frac{3}{5},\frac{1}{5},\frac{1}{5},\frac{1}{5}\right)
\end{eqnarray*}
plus vertices of the form $(v,0)$, where $v$ is a vertex of $\mP(\Lambda^3(\CC^6))$, and vertices of the form $(1,v)$, where $v$ is a vertex of $\mP(\Lambda^2(\CC^6))$. It is a matter of a straightforward calculation to show that the well-known inequalities for the spectral polytope of $\Lambda^3(\CC^7)$ \cite{Klyachko04,BD} give the same set of vertices.

The two-step procedure that we used to compute the spectral polytope for $\Lambda^3(\CC^7)$ can be generalised to $\Lambda^3(\CC^N)$ in a straightforward way. In the first step, we have
\[N_2=\CC^3\otimes\Lambda^2(\CC^{N-3}),\]
and in the second step 
\[N_2'=\left(\CC^2\otimes\CC^2\otimes\CC^{N-5}\right)\oplus\left(\CC\otimes\CC\otimes\Lambda^2\left(\CC^{N-6}\right)\right).\]
spectral polytopes of the components are known. The polytope of the second component is a properly shifted polytope of two fermions on $N-6$ levels. The polytope of the first component can be found in \cite{VW}. However, it is difficult to determine the join of the two momentum images. In order to obtain some subset of the intersection of the join with $\mkt_+$, one can compute the join of the momentum rosettes of the components, which is a straightforward task. Such a procedure simplifies for sufficiently large $N$, as all the polytopes $\mP(\CC^2\otimes\CC^2\otimes\CC^{M})$ for $M\geq 4$ are isometric \cite{VW}. This means that vertices of $\mP(\CC^2\otimes\CC^2\otimes\CC^{M+1})$ are of the form $(v,0)$, where $v$ is a vertex of $\mP(\CC^2\otimes\CC^2\otimes\CC^{M})$ for $M\geq 4$.

\subsection{Bosons}\label{bosons}
In the case of bosons, the quantum marginal problem is trivial, as for $L\geq2$ the whole momentum image is convex, namely
\[\fl\mu(S^L(\CC^N))\cap\mkt=\conv{\mu(\ket{L,0,\dots ,0}),\mu(\ket{0,L,0,\dots,0}),\dots,\mu(\ket{0,\dots, 0,L})}.\]
The momentum image is a $N$-simplex. The spectral polytope is the intersection of the simplex with the positive Weyl chamber. This intersection is the cone at the highest weight intersected with $\mkt_+$. The cone at the highest weight is generated by rays from the highest weight to the weights from $N_2$:
\[(1-t)\lambda+t\mu(\ket{L-2,0,\dots,0,2,0,\dots,0}),\ t\geq0.\]
The intersection of the cone with $\mkt_+$ takes place for $t\geq1$, i.e. outside the convex hull of $\lambda$ and $\Lambda_2$, but the spectral polytope is still equal to $\mc{P}_2$. This is because the entire momentum image is convex.

\subsection{Fermionic Fock space}
As we explained in section \ref{sec:preliminaries}, the $Spin(2N)$ representation on the Fock space has two irreducible components. Therefore, we will consider two momentum maps, each stemming from an irreducible component. The stabiliser of the highest weight is in $\mc{F}_e$ is $K_\lambda=U(N)$, where the positive root operators of $K_\lambda$ are of the form $a_ia_j^\dagger$, $1\leq i<j\leq N$. The $N_j$ spaces are the components with the number of particles equal to $2j$, i.e.
\[\fl N_j(\mc{F}_e)=\Lambda^{2j}(\CC^N)=\SpC{\ket{a_1}\wedge\ket{a_2}\wedge\dots\wedge\ket{a_{2j}}:\ 1\leq a_1<\dots<a_{2j}\leq N}.\]
For $\mc{F}_o$, the stabiliser of the highest weight is $K_\lambda=U(N)$, with the positive root operators of the form $a_ia_j^\dagger$, $1<i<j\leq N-1$ and $a_ia_N$, $1\leq i\leq N-1$. The $N_j$ spaces are
\begin{eqnarray*}
N_j(\mc{F}_o)\cong \Lambda^{2j-1}(\CC^{N-1})\oplus\Lambda^{2j}(\CC^{N-1})=\\=\SpC{\ket{a_1}\wedge\dots\wedge\ket{a_{2j-1}}:1\leq a_1<\dots<a_{2j-1}\leq N-1}\oplus \\ \oplus\ \SpC{\ket{a_1}\wedge\dots\wedge\ket{a_{2j}}\wedge\ket{N}:1\leq a_1<\dots<a_{2j}\leq N-1}.
\end{eqnarray*}

The representations with $N\leq 5$ are spherical, so the spectral polytope is equal to the polytope from the doubly excited states. The case $N=1$ is trivial, as the components are one-dimensional, i.e. they consist of the highest weight spaces, and their momentum images are single points $\mu(\klambda)=\lambda$. For $N=2$ we have
\[\mc{F}_e=\SpC{\vac}\oplus\SpC{a_1^\dagger a_2^\dagger\vac}\cong \CC\oplus\CC.\]
There are two weight vectors, who are root-neighbours, which means that the line segment connecting the two weights is not in the momentum image. This means that $\mu(\mc{F}_e)\cap\mkt_+=\left(\frac{1}{2},\frac{1}{2}\right)$. Similarly, $\mu(\mc{F}_o)\cap\mkt_+=\left(\frac{1}{2},-\frac{1}{2}\right)$. For $N=3$, we have
\begin{eqnarray*}
\fl \mc{F}_e=\SpC{\vac,\ket{1}\wedge\ket{2},\ket{1}\wedge\ket{3},\ket{2}\wedge\ket{3}},\ \ 
\mc{F}_o=\SpC{\ket{1},\ket{2},\ket{3},\ket{1}\wedge\ket{2}\wedge\ket{3}}.
\end{eqnarray*}
Again, any two of the above weight vectors are root-neighbours, so the momentum polytopes are single points, i.e. $\mu(\mc{F}_e)\cap\mkt_+=\left(\frac{1}{2},\frac{1}{2},\frac{1}{2}\right)$, $\mu(\mc{F}_o)\cap\mkt_+=\left(\frac{1}{2},\frac{1}{2},-\frac{1}{2}\right)$. The case with $N=4$ is the first case, where spaces $N_2$ are nontrivial. 
\[N_2(\mc{F}_e)=\SpC{\ket{1}\wedge\ket{2}\wedge\ket{3}\wedge\ket{4}},\ N_2(\mc{F}_o)=\SpC{\ket{1}\wedge\ket{2}\wedge\ket{3}}.\]
Spaces $N_2$ are one-dimensional, hence the spectral polytopes are just the line segments.
\begin{eqnarray*}
\fl \mu(\mc{F}_e)\cap\mkt_+=\left\{(1-t)\left(\frac{1}{2},\frac{1}{2},\frac{1}{2},\frac{1}{2}\right)+t\left(-\frac{1}{2},-\frac{1}{2},-\frac{1}{2},-\frac{1}{2}\right):\ 0\leq t\leq\frac{1}{2}\right\}, \\
\fl \mu(\mc{F}_e)\cap\mkt_+=\left\{(1-t)\left(\frac{1}{2},\frac{1}{2},\frac{1}{2},-\frac{1}{2}\right)+t\left(-\frac{1}{2},-\frac{1}{2},-\frac{1}{2},\frac{1}{2}\right):\ 0\leq t\leq\frac{1}{2}\right\},
\end{eqnarray*}
where the range of $t$ is such that the line segment is contained in $\mkt_+$. Equivalently, the spectral polytopes are given by equations
\begin{eqnarray*}
\fl \mP(\mc{F}_e)=\left\{(\eta_1,\dots,\eta_4)\in\RR^4:\ \eta_1=\eta_2=\eta_{3}=\eta_4,\ 1\leq\eta_1\leq\frac{1}{2}\right\}, \\
\fl \mP(\mc{F}_o)=\left\{(\eta_1,\dots,\eta_4)\in\RR^4:\ \eta_1=\eta_2=\eta_{3},\ \eta_4=-\eta_{3},\ 1\leq\eta_1\leq\frac{1}{2}\right\}.
\end{eqnarray*}
For $N=5$, the $N_2$-spaces are 
\[N_2(\mc{F}_e)=\Lambda^4(\CC^5)\cong\CC^5,\ N_2(\mc{F}_o)=\Lambda^3(\CC^4)\oplus\Lambda^4(\CC^4)\cong\CC^4\oplus\CC.\]
It is straightforward to check that in both cases any two weights are root-neighbours. Therefore, the $K_\lambda$-momentum images of $N_2$ are the highest weights for the respective representations of $K_\lambda$, i.e.
\begin{eqnarray*}
\mu_{K_\lambda}(N_2(\mc{F}_e))\cap(\mkt_\lambda)_+=\left(\frac{1}{2},-\frac{1}{2},-\frac{1}{2},-\frac{1}{2},-\frac{1}{2}\right),\\
\mu_{K_\lambda}(N_2(\mc{F}_o))\cap(\mkt_\lambda)_+=\left(\frac{1}{2},-\frac{1}{2},-\frac{1}{2},-\frac{1}{2},\frac{1}{2}\right).
\end{eqnarray*}
Therefore, the spectral polytopes are the line segments from the highest weight to the momentum image of $N_2$, intersected with $\mkt_+$. The results read
\begin{eqnarray*}
\fl \mP(\mc{F}_e)=\left\{(\eta_1,\dots,\eta_5)\in\RR^5:\ \eta_1=\frac{1}{2},\ \eta_2=\eta_3=\eta_{4}=\eta_5,\ 1\leq\eta_2\leq\frac{1}{2}\right\}, \\
\fl \mP(\mc{F}_o)=\left\{(\eta_1,\dots,\eta_5)\in\RR^5:\ \eta_1=\frac{1}{2},\ \eta_2=\eta_3=\eta_{4},\ \eta_5=-\eta_{4},\ 1\leq\eta_2\leq\frac{1}{2}\right\}.
\end{eqnarray*}
The spin representations of $Spin(2N)$ are not spherical for $N\geq 6$. Hence, the polytope from doubly excited states may not be equal to the entire spectral polytope. However, the $\mP_2$ polytopes for $N=6$ and $N=7$ are easy to compute, as they are the convex hulls of some known polytopes for the scenarios with a fixed number of fermions and the highest weight. The question of whether in these cases we have $\mP_2=\mP$ remains open.

\section{Summary and outlook}
In this work, we describe the set of quantum marginals, which are the one-particle reduced density matrices stemming from a pure state. By the local unitary symmetry of the problem, the set of all one-particle reduced density matrices that are compatible with some pure quantum state is fully described by their spectra. Notably, the set of admissible spectra forms a polytope, which is called the spectral polytope. Formulating the problem of the description of the spectral polytope in terms of representation theory of compact Lie groups allows us to treat a variety of physical scenarios in a consistent way.  In particular, we applied our methodology to the description of the set of correlation matrices in the fermionic Fock space. Using the local description of the spectral polytope around its distinguished vertex (the highest weight), we gave a general construction of two polytopes that bound the spectral polytope. The polytope, which is the upper bound involves all the $j$-excitation spaces. Direct 
computation of this polytope is of interest in quantum chemistry \cite{benavides2016natural,benavidesharmonium}, as the facets of the polytope bound the spectral polytope sharply, hence they yield some generalised Pauli constraints for systems of $L$ fermions in $N$ modes. The generalised Pauli constraints give rise to the phenomenon of pinning and quasipinning of fermionic occupation numbers \cite{schilling2013natural,schilling2015hubbard,benavidespinning1,benavidespinning2,tennie2016pinning1,tennie2016pinning2,tennie2017influence}. These constraints are expected to be easier to compute than the constraints for the entire spectral polytope. To this end, it is crucial to develop new tools that allow one to compute the {\it joins} of momentum maps (see Definition \ref{def:join}). The construction of the polytope being the lower bound involves only a small part of the entire Hilbert space, namely the doubly excited states and its description is more tractable. We classified quantum systems, where both bounds 
coincide giving the whole spectral polytope. In particular, we obtained new results regarding the spectral polytopes in the fermionic Fock space. As the examples computed in Section \ref{sec:computational_results} show, the spectral polytopes of $j$-excitation spaces have an inductive structure in the sense, that they correspond to smaller quantum systems of analogous type. In particular, in order to compute the local cone at the ground state for distinguishable particles, one needs the knowledge of the spectral polytopes for systems of distinguishable particles with a smaller number of modes. A similar phenomenon is true for systems of fermions. It would be interesting to explore further this behaviour to obtain other new bounds for the spectral polytope.

Another way to develop further the methodology introduced in this paper would be the description of the local cones around other vertices of the spectral polytope. We expect that the spectral polytope is the intersection of a small number of cones, while the remaining vertices come from the intersections of the cones.

The results of this paper are also useful for finding new classes of maximally entangled states. In the SLOCC classification the maximally entangled states are the states that are mapped by the momentum map to the vertex of the positive Weyl chamber. In all scenarios considered in this paper the polytope from doubly excited states contains the vertex of the positive Weyl chamber. Therefore, there are some classes of doubly excited states that are maximally entangled in the sense of the said SLOCC classification.

\ack
We would like to thank Marek Ku\'s for encouragement and Adam Sawicki and Micha\l\ Oszmaniec for many fruitful discussions. We thank Peter Heinzner for the discussions during our stay in Bochum in August 2016. TM would like to thank the University of G\"{o}ttingen for hospitality during his stay in the summer of 2016, where this work was partially done. TM was supported by DAAD Short-Term Reserach Grant no. 57214227, Polish Ministry of Science and Higher Education ``Diamentowy Grant'' no. DI2013 016543 and ERC grant QOLAPS. V.V.T. is supported by the DFG grant Sachbeihilfe DFG-AZ: TS 352/1-1.


\section*{References}
\begin{thebibliography}{11}
\bibitem{raport} Stillinger F H et al 1995 Mathematical challenges from theoretical/computational chemistry
(Washington: National Academy Press)

\bibitem{GS} Guillemin, V., Sternberg, S., Symplectic Techniques in Physics, Cambridge University Press (1984)

\bibitem{Horodecki} Horodecki, R., Horodecki, P., Horodecki, M., Horodecki, K., Quantum entanglement Rev. Mod. Phys. Vol. 81, No. 2, pp. 865-942, 2009

\bibitem{Atiyah} Atiyah, M., F., Convexity and commuting Hamiltonians, Bull. London Math. Soc. 14, 1-15 (1982)

\bibitem{Kirwan} F. C. Kirwan Cohomology quotients in symplectic and algebraic geometry, Mathematical Notes, Vol. 31, Princeton Univ. Press, Princeton (1984)

\bibitem{polytopes} Walter, M., Doran, B., Gross, D., Christandl, M. Entanglement polytopes: Multiparticle Entanglement
from Single-Particle Information, Science 340 (6137), 1205-1208, 2013

\bibitem{Sudberry} Higuchi, A., Sudbery, A., Szulc, J. , One-qubit reduced states of a pure many-qubit state: polygon
inequalities, Phys. Rev. Lett. 90, 107902, 2003

\bibitem{MOS} Maci\k{a}\.{z}ek, T., Oszmaniec, M., Sawicki, A., {\it How many invariant polynomials are needed to decide local unitary equivalence of qubit states?}, J. Math. Phys. 54, 2013

\bibitem{Klyachko04} Klyachko, A., {\it Quantum marginal problem and representations of the symmetric group}, arXiv:quant-ph/0409113 (2004)

\bibitem{Klyachko05} Klyachko, A., {\it Quantum marginal problem and N-representability}, arXiv:quant-ph/0511102 (2005)

\bibitem{Klyachko09} Klyachko, A., {\it The Pauli exclusion priciple and beyond}, preprint, quant- ph/0904.2009 (2009)

\bibitem{AK08} Altunbulak, M., Klyachko, A., {\it The Pauli principle revisited}, Communications in Mathematical Physics Volume 282, Issue 2, pp 287-322 (2008)

\bibitem{CDKW} Christandl, M., Doran, B., Kousidis, S., Walter, M., {\it Eigenvalue distributions of reduced density matrices} Commun. Math. Phys. 332, 1-52 (2014)

\bibitem{Ruskai70} Ruskai, M. B.: N -representability problem: Particle-hole equivalence. J. Math. Phys. 11, 3218?3224 (1970)

\bibitem{Ruskai07} Ruskai, M.B.: Connecting N-representability to Weyl's problem: The one particle density matrix for N = 3 and R = 6. J. Phys. A: Math. Theor. 40, F961-F967 (2007)

\bibitem{Kirwan84} Kirwan, F., C., {\it Convexity properties of the momentum mapping, III} Invent. Math., 77, pp. 547?552 (1984)

\bibitem{GS82} Guillemin, V., Sternberg, S., {\it Convexity properties of the momentum mapping}, Invent. Math., Volume 67, Issue 3, pp 491-513 (1982)

\bibitem{Sjamaar} Sjamaar, R., {\it Convexity Properties of the momentum Mapping Re-examined}, Advances in Mathematics Volume 138, Issue 1, Pages 46-91 (1998)

\bibitem{Brion} Brion, M., {\it On the general faces of the momentum polytope}, Int. Math. Res. Not., 185?201. MR1677271 (2000i:14068) (1999)

\bibitem{VW} Vergne, M., Walter, M., {\it Inequalities for Moment Cones of Finite-Dimensional Representations}, arXiv:1410.8144 (2014)

\bibitem{Smirnov} Smirnov, A., V., {\it Decomposition of symmetric powers of irreducible representations of semisimple Lie algebras and the Brion polytope}, Trans. Moscow Math. Soc., Pages 213?234 S 0077-1554(04)00143-8 (2004)

\bibitem{GKM} Grabowski, J., Ku\'{s}, M., Marmo, G., {\it Geometry of quantum systems: density states and entanglement}, J.Phys. A38 10217-10244 (2005) 

\bibitem{SHK} Sawicki, A., Huckleberry, A., Ku\'{s}, M. {\it Symplectic geometry of entanglement}, Comm. Math. Phys. 305,
441?468 (2011)

\bibitem{Wildberger} N. J. Wildberger, {\it The momentum map of a Lie group representation}, Trans. AMS {\bf 330} (1992), 257--268.

\bibitem{LM} Landsberg, J., M., Manivel, L., {\it On the projective geometry of rational homogeneous varieties}, Comment. Math. Helv. 78, 65-100 0010-2571/03/010065-36 (2003)

\bibitem{BD} Borland, R. E., Dennis K., {\it The conditions on the one-matrix for three-body fermion wavefunctions with one-rank equal to six}. J. Phys. B, 5:7-15 (1972)

\bibitem{MichalPhD} Oszmaniec, M., {\it Applications of differential geometry and representation theory to description of quantum correlations}, PhD thesis, arXiv:1412.4657 (2014)

\bibitem{MCT13} de Melo, F., \'{C}wikli\'{n}ski, P., , Terhal, B., M., {\it The Power of Noisy Fermionic Quantum Computation}, New J. Phys. 15 013015 (2013)

\bibitem{atiyah} Atiyah, M., F. {\textit Convexity and commuting Hamiltonians}, Bull. London Math. Soc. 14, 1-15 (1982)

\bibitem{K84} Kirwan, F. C. \textit{Convexity properties of the moment mapping},
III, Invent. Math. 77, 547552. (1984)

\bibitem{Brion1987} Brion, M., {\it Sur l'image de l'application moment}. In: Malliavin MP. (eds) Séminaire d'Algèbre Paul Dubreil et Marie-Paule Malliavin. Lecture Notes in Mathematics, vol 1296, (1987)

\bibitem{Brion} Brion, M., {\it On the general faces of the moment polytope} Int. Math. Res. Not., 4:185?201 (1999)

\bibitem{Mumford} Ness, L., and Mumford, D., {\it A Stratification of the Null Cone Via the Moment Map}  Amer. J. Math., 106:1281?1329 (1984)

\bibitem{BS} Berenstein, A., and Sjamaar, R., {\it Coadjoint orbits, moment polytopes, and the Hilbert-Mumford criterion} J. Am. Math. Soc., 13:433-466 (2000)

\bibitem{WDGC12}Walter, M., Doran, B., Gross, D., Christandl, M. \textit{Entanglement polytopes: Multiparticle Entanglement from Single-Particle Information}, Science 340 (6137), 1205-1208, 2013

\bibitem{SOK12} Sawicki, A., Oszmaniec, M., Ku\'s, M. \textit{Critical sets of the total variance can detect all stochastic local operations and classical communication classes of multiparticle entanglement} Phys.
Rev. A 86, 040304(R), 2012

\bibitem{SOK14} Sawicki, A., Oszmaniec, M., Ku\'s, M. \textit{Convexity of momentum map, Morse index, and quantum entanglement},  Rev. Math. Phys. 26, 1450004, 2014

\bibitem{MS15} Maci\k{a}\.{z}ek, T., Sawicki, A., {\it Critical points of the linear entropy for pure L-qubit states},  J. Phys. A: Math. Theor. 48 045305, DOI: 10.1088/1751-8113/48/4/045305, (2015)

\bibitem{HH96}Heinzner, P., Huckleberry, A.,
\newblock K\"{a}hlerian potentials and convexity properties of the moment map,
\newblock  {Invent. Math.} \textbf{126}, 6584, 1996.

\bibitem{HSS03}Higuchi, A., Sudbery, A., and Szulc, J.,
\newblock One-qubit reduced states of a pure many-qubit state: polygon inequalities,
\newblock  {Phys. Rev. Lett.} \textbf{90}, 107902, 2003.

\bibitem{LCV} Liu, Y.-K., Christandl, M., and Verstraete, F.,. {\it Quantum Computational Complexity of the N -Representability Problem: QMA Complete}, Phys. Rev. Lett., 98:110503, (2007)

\bibitem{Liu} Liu, Y.-K., {\it Consistency of Local Density Matrices is QMA-complete} In Proc. RANDOM, pages 438?449 (2006)

\bibitem{Pauli} Pauli, W., {\it \"{U}ber den Zusammenhang des Abschlusses der Elektronengruppen im Atom mit der Komplexstruktur der Spektren}, Z. Phys., 31:765-783 (1925)

\bibitem{Ness} Ness, L., {\it A stratification of the null cone via the moment map [with an appendix by D. Mumford]}, Amer. J. Math 106(6), pp.1281-1329 (1984)

\bibitem{Bravyi} Bravyi, S., {\it Universal Quantum Computation with the nu=5/2 Fractional Quantum Hall State}, Phys. Rev. A 73, 042313 (2006)

\bibitem{OGK14} Oszmaniec, M., Gutt, J., Ku\'{s}, M., {\it Classical simulation of fermionic linear optics augmented with noisy ancillas}, Phys. Rev. A 90, 020302(R) (2014)

\bibitem{OK14} Oszmaniec, M., Ku\'{s}, M., {\it Fraction of isospectral states exhibiting quantum correlations}, Phys. Rev. A 90, 010302(R) (2014)

\bibitem{OK13} Oszmaniec, M., Ku\'{s}, M., {\it A universal framework for entanglement detection}, 	Phys. Rev. A 88, 052328 (2013)

\bibitem{SL14} S\'{a}rosi, G., L\'{e}vay, P., {\it Entanglement in fermionic Fock space}, J. Phys. A: Math. Theor. 47 115304 (2014)

\bibitem{LH14} L\'{e}vay, P., Holweck, F., {\it Embedding qubits into fermionic Fock space, peculiarities of the four-qubit case}, Phys. Rev. D 91, 125029 (2015)

\bibitem{Kac-nilp-orb} Kac, V., {\it Some remarks on nilpotent orbits}, J. of Algebra {\bf 64}, 190213 (1980)

\bibitem{HKS13} Huckleberry, A., Ku\'{s}, M., Sawicki, A., {\it Bipartite entanglement, spherical actions and geometry of local unitary orbits}, J. Math. Phys. 54, 022202 (2013)

\bibitem{Broe-TomDie} Br\"ocker, Th., tom Dieck, T., {\it Representations of Compact Lie Groups}, Springer, New-York, 1985. 

\bibitem{Humphreys} Humphreys, J., {\it Introduction to Lie Algebras and Representation Theory}, Springer-Verlag, 1972.

\bibitem{Knop-MultFreeSpa} Knop, F., {\it Some remarks on multiplicity free spaces}, Proc. NATO Adv. Study Inst. on Representation Theory and Algebraic Geometry (A. Broer, G. Sabidussi, eds.), Nato ASI Series C, Vol. {\bf 514}, Dortrecht: Kluwer (1998)

\bibitem{SGC13} Schilling, C., Gross, D, Christandl, M., {\it Pinning of Fermionic Occupation Numbers}, Phys. Rev. Lett. 110, 040404 (2013)

\bibitem{S15} Schilling, C., {\it Quasipinning and its relevance for N-fermion quantum states}, Phys. Rev. A 91, 022105 (2015)

\bibitem{BCMW15} B\"{u}rgisser, P., Christandl, M., Mulmuley, K., D., Walter, M., {\it Membership in moment polytopes is in NP and coNP}, SIAM J. Comput., 46 (3), 972-991, 2017

\bibitem{KS82} Kostant, B., Sternberg, B., {\it Symplectic projective orbits. In: New directions in applied mathematics, papers presented April 25/26, 1980, on the occasion of the Case Centennial Celebration}, New York: Springer, pp 81-84 (1982)

\bibitem{chevalley} Chevalley, C., {\it The algebraic Theory of Spinors}, Columbia University Press, 1954

\bibitem{universal-fermions} Lin, C., Jianxin, C., D., Z., Djokovic, Bei Z., {\it Universal Subspaces for Local Unitary Groups of Fermionic Systems}, Communications in Mathematical Physics, Volume 333, Issue 2, pp 541-563, 2015

\bibitem{benavides2016natural} Benavides-Riveros, C.L. and Schilling, C., {\it Natural Extension of Hartree--Fock Through Extremal 1-Fermion Information: Overview and Application to the {L}ithium Atom}, Z. Phys. Chem., 230(5-7), 703-717, 2016

\bibitem{benavidesharmonium} Benavides-Riveros, C. L. and Gracia-Bond\'{\i}a, J. M. and V\'arilly, J. C. {\it The lowest excited configuration of harmonium}, Phys. Rev. A 86, 022525, 2012

\bibitem{schilling2013natural} Schilling, C., {\it Natural orbitals and occupation numbers for harmonium: Fermions versus bosons}, Phys. Rev. A, vol. 88, no 4, pp 042105, 2013

\bibitem{schilling2015hubbard} Schilling, C., {\it Hubbard model: Pinning of occupation numbers and role of symmetries}, Phys. Rev. B 92, 155149, 2015

\bibitem{benavidespinning1} Benavides-Riveros, Carlos L. and Gracia-Bond\'{\i}a, Jos\'e M. and Springborg, Michael, {\it Quasipinning and entanglement in the lithium isoelectronic series}, Phys. Rev. A 88, 022508, 2013
\bibitem{benavidespinning2} Benavides-Riveros, Carlos L. and Springborg, Michael, {\it Quasipinning and selection rules for excitations in atoms and molecules}, Phys. Rev. A 92, 012512, 2015

\bibitem{tennie2016pinning1} Tennie, F. and Ebler, D. and Vedral, V. and Schilling, C., {\it Pinning of fermionic occupation numbers: General concepts and one spatial dimension}, Phys. Rev. A, 93(4), 042126, 2016

 \bibitem{tennie2016pinning2} Tennie, F. and Vedral, V. and Schilling, C., {\it Pinning of fermionic occupation numbers: Higher spatial dimensions and spin}, Phys. Rev. A, 94(1), 012120, 2016
  
\bibitem{tennie2017influence} Tennie, F. and Vedral, V. and Schilling, C., {\it Influence of the fermionic exchange symmetry beyond {P}auli's exclusion principle}, Phys. Rev. A, 95(2), 022336, 2017

\end{thebibliography}
\end{document}